\def\doublecol{0}

\if \doublecol1
    \documentclass[journal]{IEEEtran}
\else
    \documentclass[journal, onecolumn]{IEEEtran}
\fi

\usepackage{amsmath, bm, amssymb, algorithm, algpseudocode, mathtools, svg,ifluatex, pgfplots, amsthm, caption, subcaption,bookmark,hyperref,cite,multirow,url,tikz,bbm, orcidlink, enumitem,soul}
\usetikzlibrary{graphs, positioning, quotes, shapes.geometric}
\usepackage[acronym,shortcuts]{glossaries}
\allowdisplaybreaks

\pgfplotsset{compat=1.15}

\newcommand{\br}[1]{\left( #1 \right)}
\newcommand{\brsq}[1]{\left[ #1 \right]}
\newcommand{\brcur}[1]{\left\{ #1 \right\}}

\newcommand{\expc}[2]{\mathbb{E}_{#1}\brsq{#2}}
\newcommand{\expcs}[2]{\mathbb{E}_{#1}[#2]}
\newcommand{\var}[2]{\mathrm{Var}_{#1}[#2]}

\newcommand{\pr}[1]{\mathrm{Pr}\brcur{#1}}

\newcommand{\hS}{\hat{S}}
\newcommand{\hM}{\hat{M}}

\newcommand{\calS}{\mathcal{S}}
\newcommand{\calR}{\mathcal{R}}

\newcommand{\calC}{\mathcal{C}}
\newcommand{\calO}{\mathcal{O}}
\newcommand{\calX}{\mathcal{X}}
\newcommand{\calN}{\mathcal{N}}

\newcommand{\calY}{\mathcal{Y}}
\newcommand{\calZ}{\mathcal{Z}}
\newcommand{\calU}{\mathcal{U}}
\newcommand{\calV}{\mathcal{V}}
\newcommand{\calT}{\mathcal{T}}
\newcommand{\calW}{\mathcal{W}}
\newcommand{\calM}{\mathcal{M}}

\newcommand{\calP}{\mathcal{P}}

\newcommand{\typset}[2]{\calT^{(n)}_{\epsilon_{#1}}({#2})}

\newcommand{\SEQ}{\mathrm{seq}}
\newcommand{\BC}{\mathrm{BC}}
\newcommand{\MAC}{\mathrm{MAC}}

\newcommand{\C}{\mathrm{C}}

\newcommand{\Bern}{\mathrm{Bern}}

\newcommand{\pddf}[2]{\calP_{#1, #2}}
\newcommand{\pdd}{\pddf{D_1}{D_2}}

\DeclareMathOperator*{\argmin}{arg\,min}

\newacronym{mimo}{MIMO}{multiple-input multiple-output}
\newacronym{mu}{MU}{multi-user}
\newacronym{ofdm}{OFDM}{orthogonal frequency-division multiplexing}
\newacronym{5g}{5G}{fifth generation}
\newacronym{6g}{6G}{sixth generation}
\newacronym{b5g}{B5G}{Beyond 5G}

\newacronym{crlb}{CRB}{Cram\'er-Rao bound}
\newacronym{bcrb}{BCRB}{Bayesian Cram\'er-Rao bound}
\newacronym{bfim}{BFIM}{Bayesian Fisher information matrix}
\newacronym{repms}{REPMS}{Riemannian Exact Penalty Method via Smoothing}
\newacronym{srepms}{SREPMS}{Stochastic Riemannian Exact Penalty Method via Smoothing}
\newacronym{sdr}{SDR}{semidefinite relaxation}
\newacronym{sdp}{SDP}{semidefinite programming}
\newacronym{mmwave}{mmWave}{millimeter wave}
\newacronym{isac}{ISAC}{integrated sensing and communications}
\newacronym{jsmc}{JSMC}{joint state and message communications}
\newacronym{aoa}{AoA}{angle of arrival}
\newacronym{aod}{AoD}{angle of departure}
\newacronym{qos}{QoS}{Quality of Service}
\newacronym{th}{THz}{Terahertz}
\newacronym{dof}{DoFs}{degrees of freedom}
\newacronym{iot}{IoT}{Internet of Things}
\newacronym{uav}{UAV}{unmanned aerial vehicle}
\newacronym{tx}{Tx}{transmitter}
\newacronym{rx}{Rx}{receiver}
\newacronym{ula}{ULA}{uniform linear array}
\newacronym{dft}{DFT}{Discrete Fourier Transform}
\newacronym{re}{RE}{resource element}
\newacronym{dl}{DL}{Downlink}
\newacronym{ul}{UL}{Uplink}
\newacronym{psd}{PSD}{positive semidefinite}
\newacronym{mse}{MSE}{mean squared error}
\newacronym{svd}{SVD}{singular value decomposition}
\newacronym{cpu}{CPU}{central processing unit}
\newacronym{mui}{MUI}{multi-user interference}
\newacronym{mi}{MI}{mutual information}
\newacronym{srmo}{SRMO}{stochastic Riemannian manifold optimization}
\newacronym{srgd}{SRGD}{stochastic Riemannian gradient descent}
\newacronym{srcg}{SRCG}{stochastic Riemannian conjugate gradient}
\newacronym{kkt}{KKT}{Karush–Kuhn–Tucker}
\newacronym{as}{a.s.}{almost sure}
\newacronym{papr}{PAPR}{peak to average power ratio}
\newacronym{rd}{R-D}{rate-distortion}
\newacronym{map}{MAP}{maximum a posteriori}
\newacronym{mmse}{MMSE}{minimum mean squared error}
\newacronym{iid}{i.i.d}{independently and identically distributed}
\newacronym{dmc}{DMC}{discrete memoryless channel}
\newacronym{sddmc}{SD-DMC}{state-dependent discrete memoryless channel}
\newacronym{sddmbc}{SD-DMDBC}{state-dependent discrete memoryless degraded broadcast channel}
\newacronym{csi}{CSI}{channel state information}
\newacronym{csit}{CSIT}{channel state information at transmitter}
\newacronym{csir}{CSIR}{channel state information at receiver}
\newacronym{bc}{BC}{broadcast channel}
\newacronym{simo}{SIMO}{single-input-multiple-output}
\newacronym{bcd}{BCD}{block coordinate descent}
\newacronym{ba}{BA}{Blahut-Arimoto}
\newacronym{bs}{BS}{base station}
\newacronym{ue}{UE}{user equipment}

\newacronym{si}{SI}{side information}
\newacronym{sit}{SI-T}{side information at transmitter}
\newacronym{sir}{SI-R}{side information at receiver}
\newacronym{lln}{LLN}{law of large numbers}
\newacronym{cd}{C-D}{capacity-distortion}
\newacronym{sddmmac}{SD-DMMAC}{state-dependent discrete memoryless multiple access channel}

\newtheorem{theorem}{Theorem}
\newtheorem{lemma}{Lemma}

\newtheorem{prop}{Proposition}

\theoremstyle{definition}
\newtheorem{definition}{Definition}

\newtheorem{remark}{Remark}

\glsdisablehyper
\allowdisplaybreaks

\tikzset{%
    block-common/.style={draw, fill=white, minimum height=1.5em, minimum width=4em},
    block/.style={rectangle, block-common},
    smallblock/.style={draw, fill=white, minimum height=1em, minimum width=1em},
    bigblock/.style={block, minimum height=8em},
    txtblock/.style={block, align=center, minimum height=4em},
    txtbigblock/.style={bigblock, align=center},
    input/.style={inner sep=1pt},       
    output/.style={inner sep=1pt},      
    sum/.style = {draw, fill=white, circle, minimum size=1.1em, inner sep=0pt,
      font={\small$+$}},
    prod/.style = {draw, fill=white, circle, minimum size=1.1em, inner sep=0pt,
      font={\normalsize$\times$}},
    pinstyle/.style = {pin edge={to-,thin,black}}
}

\usepackage{cite}

\begin{document}
% \title{A Unified Capacity-Distortion Framework for Integrated Sensing and Communications}
% \title{Information-Theoretic Analysis of Capacity-Distortion Trade-Offs in ISAC Systems}
% \title{An Analysis of Capacity-Distortion Trade-Offs in Memoryless ISAC Systems}
% \title{On Multi-user Joint State and\\ Message Communications}
\title{A Simultaneous Decoding Approach to Joint State and Message Communications}

\author{Xinyang~Li \orcidlink{0000-0001-7262-5948},~\IEEEmembership{Student Member,~IEEE,}
        Yiqi~Chen \orcidlink{0000-0002-4850-2072},~\IEEEmembership{Member,~IEEE,}\\
        Vlad~C.~Andrei \orcidlink{0000-0001-5443-0100},~\IEEEmembership{Student Member,~IEEE,}
        Aladin~Djuhera \orcidlink{0009-0005-1641-8801},~\IEEEmembership{Student Member,~IEEE,}\\
        Ullrich~J.~M\"onich \orcidlink{0000-0002-2390-7524},~\IEEEmembership{Senior Member,~IEEE,}
        and~Holger~Boche \orcidlink{0000-0002-8375-8946},~\IEEEmembership{Fellow,~IEEE}% <-this % stops a space
\thanks{The authors are with the Department
of Electrical and Computer Engineering, Technical University of Munich, Munich, 80333 Germany (e-mail: \{xinyang.li, yiqi.chen, vlad.andrei, aladin.djuhera, moenich, boche\}@tum.de).}% <-this % stops a space
% \thanks{Manuscript received April 19, 2005; revised August 26, 2015.}
}

% \markboth{Journal of \LaTeX\ Class Files,~Vol.~14, No.~8, August~2015}%
% {Shell \MakeLowercase{\textit{et al.}}: Bare Demo of IEEEtran.cls for IEEE Journals}

\maketitle

\begin{abstract}
The \ac{cd} trade-offs for \ac{jsmc} over single- and multi-user channels are investigated, where the transmitters have access to generalized state information and feedback while the receivers jointly decode the messages and estimate the channel state. 
A coding scheme is proposed based on backward simultaneous decoding of messages and compressed state descriptions without the need for the Wyner-Ziv random binning technique. For the point-to-point channel, the proposed scheme results in the optimal \ac{cd} function. For the \ac{sddmbc}, the successive refinement method is adopted for designing multi-stage state descriptions. With the simultaneous decoding approach, the derived achievable region is shown to be larger than the region obtained by the sequential decoding approach that is utilized in existing works. As for the \ac{sddmmac}, in addition to the proposed method, Willem's coding strategy is applied to enable partial collaboration between transmitters through the feedback links. Moreover, the state descriptions are shown to enhance both communication and state estimation performance.
Examples are provided for the derived results to verify the analysis, either numerically or analytically. With particular focus, simple but representative \ac{isac} systems are also considered, and their fundamental performance limits are studied.
\end{abstract}

\begin{IEEEkeywords}
Joint state and message communications, capacity-distortion trade-off, simultaneous decoding. %,integrated sensing and communications.
\end{IEEEkeywords}

% \fi
\IEEEpeerreviewmaketitle

\glsresetall

% \section{Plans}
% \input{plans}

\section{Introduction}\label{sec:intro}
% for IEEE journal papers produced under \LaTeX\ using
% IEEEtran.cls version 1.8b and later\cite{ahmadipour2022information}.
% \subsection{Background and Related Works}

In \ac{jsmc}, the goal is not only to transmit messages reliably but also to estimate the channel state accurately over a \ac{sddmc}. It is motivated by the growing demand for integrated systems that unify sensing and communications\cite{liu2022survey}, particularly in emerging applications such as the Internet of Things\cite{nguyen20216g} and autonomous driving\cite{wang2018networking}. For instance, in smart factories, wireless access points equipped with additional sensors such as cameras can assist a warehouse robot in localizing itself by encoding sensor information into messages. Similarly, in cellular networks, a \ac{bs} sends downlink signals to a \ac{ue} while simultaneously estimating the state of a target via the echo signals. 

The information-theoretic limits of \ac{jsmc} can be analyzed from the perspective of the \ac{cd} trade-off, describing the relationship between the maximum achievable data rate and state estimation distortion. 
Existing works\cite{sutivong2002rate, kim2008state, zhang2011joint, choudhuri2013causal, bross2017rate, bross2020message} focus on different problem settings and most of them concentrate on single user cases. 
In\cite{sutivong2002rate} and \cite{kim2008state}, the authors consider point-to-point scenarios where the channel state information is noncausally available at the transmitter. Subsequently, the authors in\cite{zhang2011joint} explore the \ac{cd} function in the absence of the state information. The scenarios of strictly causal and causal state information at the transmitter are studied in\cite{choudhuri2013causal}, and the impact of additional channel feedback is analyzed in\cite{bross2017rate}. An extension to the two-user degraded broadcast channel is carried out in~\cite{bross2020message}, where the stronger receiver performs the joint task, and the other one only decodes the messages. 

The coding scheme achieving the \ac{cd} function/region when state information is strictly causally or causally available at the encoder\cite{choudhuri2013causal} is based on block Markov coding and Wyner-Ziv compression\cite{wyner1976rate}. In particular, the encoder compresses the state information from the previous transmission block using the Wyner-Ziv random binning method, and the bin index is encoded along with the fresh message into the transmit signal in the current block. The decoder, upon receiving the channel output, sequentially recovers the message with the bin index and then the finer Wyner-Ziv index within that bin. If perfect feedback is also present at the transmitter, \cite{bross2017rate} shows that there is no need to adopt random binning, and the state description can be decoded simultaneously with the message.

Based on the existing works, the studied channel model in this paper is equipped with more extended settings. Specifically, we consider a \ac{sddmc} where the transmitter has access to the generalized version of the channel state in a strictly causal manner,
modeled as strictly causal \ac{sit} generated memorylessly conditioned on the channel state, and feedback is assumed to be a deterministic function of the channel output.
Such a model includes but is not restricted to the settings of perfect/absent channel state and feedback. 
We propose a new coding framework that is based on backward simultaneous decoding and does not need the Wyner-Ziv random binning method used in existing literature. This results in the optimal \ac{cd} function for the point-to-point channel, an extension of the results in \cite{choudhuri2013causal, bross2017rate}. 
% Although the coding scheme in\cite{choudhuri2013causal} also applies to the studied case, we propose another strategy based on backward simultaneous decoding without needing Wyner-Ziv random binning to achieve the \ac{cd} function, which additionally, unlike the scheme in\cite{bross2017rate}, does not require perfect channel feedback. 
% The \ac{cd} function is then derived for the additive Gaussian channel with squared error as the distortion function and noisy observation of channel state present at the transmitter.

The proposed coding framework is then applied to the \ac{sddmbc} with the same generalized setup. Different from the model in\cite{bross2020message}, where only the stronger decoder wishes to estimate the channel state, we assume both receivers have the joint task. An achievable \ac{rd} region is derived based on the proposed scheme and successive refinement strategy\cite{steinberg2004successive}, better than the region obtained from the Wyner-Ziv encoding and sequential decoding\cite{li2024analysis}. The benefits come from the fact that the simultaneous decoding of messages and state descriptions can get rid of many redundant decoding error events, and, on the other hand, the state description for the weaker user can assist in the decoding steps of the stronger user.

Lastly, we investigate the \ac{jsmc} over a \ac{sddmmac}, with different noisy \ac{sit} and feedback available at transmitters. Willem's strategy\cite{willems1983achievable}, which enables partial collaboration between encoders through the feedback links, is combined with the proposed coding scheme to establish an achievable \ac{rd} region. It benefits from the simultaneous decoding not only between messages and state descriptions but also between shared and private messages and thus is shown to be larger than the \ac{rd} region derived in our previous work\cite{Li2411:Achievable}, which is based on the coding scheme in\cite{lapidoth2012multiple} and the sequential decoding approach. Notably, in the case of no feedback, such result reduces to the rate region in \cite{li2012multiple}. This motivates us to analyze the usages of state descriptions, 
% which jointly enhance the communication performance, as studied in\cite{lapidoth2012multiple,li2012multiple}, and are used for the estimation task.
that is, in addition to state reconstruction, they are also capable of enhancing communication performance, as studied in\cite{lapidoth2012multiple,li2012multiple}.

While the \ac{cd} function or region is known for only a limited number of scenarios, our work advances this field by proposing novel coding strategies that not only achieve the optimal \ac{cd} function in point-to-point settings but also outperform existing approaches in multi-user channels. These insights are not only theoretically significant but also serve as practical benchmarks for designing efficient transmission schemes. Along with the theoretical analysis, examples are provided for each studied scenario, including quadratic-Gaussian and binary symmetric channels. Moreover, we show that by interpreting the \ac{sit} and feedback in different ways, the studied channels are able to model different types of \ac{isac} systems. As special cases, the monostatic-downlink and monostatic-uplink \ac{isac} are considered, whose \ac{cd} functions are derived.

\section{Notation and Preliminaries}\label{sec:notations}
% Throughout this manuscript, we follow the conventional notation as in the prior standard literature\cite{el2011network,cover1999elements,kramer2008topics}. 
% If specified,
Random variables and their realizations are denoted by uppercase and lowercase letters like $X$ and $x$, whose sample spaces are their calligraphic version like $\calX$. $\varnothing$ stands for the empty set.
$X\sim P_X$ indicates that $X$ follows the distribution $P_X$. As special cases, $\calN(\mu, \sigma^2)$ stands for Gaussian distribution with mean $\mu$ and variance $\sigma^2$, while $\Bern(p)$ for Bernoulli distribution whose probability at $1$ is $p$. $\expcs{}{X}$ and $\var{}{X}$ are the expectation and variance of $X$, respectively. We use $X^n \triangleq (X_1, X_2, ..., X_n)$ and $x^n\triangleq (x_1, x_2, ..., x_n)$ to denote the sequences of random variables and their realizations of length $n$, and a subsequence $X_i^n$ stands for $(X_i, X_{i+1}, ..., X_n)$ if $i\le n$. Given $P_X$, $x^n$ being generated from $P_X^n$ implies that $x_i$ is generated \ac{iid} from $P_X$ for all $i\in [n]$ with $P_X^n(x^n) \triangleq \prod_{i=1}^n P_X(x_i)$.
For a positive integer $M$, we use the notation $[M] \triangleq \{1,2,...,M\}$. 
% The logarithm throughout this manuscript is taken to base $2$. 
The function $\mathbbm{1}\{x = y\}$ is indicator function, taking value $1$ if $x=y$ and $0$ else. $x \oplus y$ is the modulo-2 sum of $x$ and $y$.

Let $X \sim P_X$ and $x^n$ be a sequence generated \ac{iid} from $P^n_X$. Let $\epsilon >0$, a sequence $x^n$ is said to be $\epsilon$-typical if
\begin{equation}
    \left| \frac{N(a|x^n)}{n} - P_X(a) \right| \le \epsilon\cdot P_X(a), \quad \forall a \in \calX,
\end{equation}
where $N(a | x^n)$ counts the number of occurrences of $a$ in $x^n$. The set of all $\epsilon$-typical $x^n$ is called $\epsilon$-typical set with respect to $P_X$, denoted as $\typset{}{P_X}$. Similarly, the jointly typical set $\typset{}{P_{XY}}$ contains all jointly typical sequences $(x^n, y^n)$ with respect to $P_{XY}$. Given $x^n \in \typset{}{P_X}$, the conditionally typical set $\typset{}{P_{XY} | x^n}$ is defined as the set of all $y^n$ such that $(x^n, y^n) \in \typset{}{P_{XY}}$.

Given the random variables $(S, V, W) \sim P_{SVW}$, we define a function $h:\calV\times \calW \to \hat{\calS}$ to estimate $S$ based on $(V,W)$ with $\hat{\calS}$ the reconstruction set. Let the estimation error be measured by a distortion function $d: \calS \times \hat{\calS} \rightarrow [0, \infty)$, the following lemmas provide important properties for the choice of $h$.

\begin{lemma}\cite{zhang2011joint,ahmadipour2022information}\label{lemma:optest}
    Given a joint distribution $P_{SVW}$, the minimum expected distortion $\expcs{}{d(S,\hS)}$ is achieved by the optimal estimator
    % the optimal estimator of $S$ based on observations $V=v$ and $W=w$ with respect to a distortion function $d(s, \hat{s})$ is given by
    \begin{equation}
    \begin{split}
        h^*(v,w) &= \argmin_{\hat{s}\in \hat{\calS}}\expc{}{d(S,\hat{s})|V=v,W=w}\\
        &= \argmin_{\hat{s}\in \hat{\calS}} \sum_{s\in \calS} P_{S|VW}(s|v,w)d(s,\hat{s}).\label{eq:optest}
    \end{split}
    \end{equation}
\end{lemma}
\begin{proof}
    The proof follows by applying the law of total expectations on $\expcs{}{d(S,\hS)}$ conditioned on $V=v$ and $W=w$, and taking the minimum for each individual $(v,w)$.
\end{proof}

\begin{lemma}\cite{choudhuri2013causal}\label{lemma:markovest}
    Given a Markov chain $S-V-W$ and a distortion function $d(s, \hat{s})$, for every estimation function $h(v,w)$, there exists a $h'(v)$ such that
    \begin{equation}
        \expcs{}{d(S, h'(V))} \le \expcs{}{d(S, h(V,W))}.
    \end{equation}
\end{lemma}
\begin{proof}
    See \cite[Appendix A]{choudhuri2013causal}.
\end{proof}
Combining Lemma~\ref{lemma:optest}~and~\ref{lemma:markovest}, one can conclude that the optimal estimator of $S$ can only depend on $V$ if $S-V-W$ forms a Markov chain. In the following, we use $h^*$ to indicate all optimal estimators defined in the form of~\eqref{eq:optest}, irrespective of its domain and codomain, if there is no confusion.

Common distortion functions include the squared error $d(s,\hat{s}) = (s-\hat{s})^2$ and the Hamming distance $d(s,\hat{s}) = d_H(s,\hat{s}) = \mathbbm{1}\{s = \hat{s}\}$. Given a joint distribution $P_{SVW}$, the optimal estimator for the squared error based on observation $V=v$ and $W=w$ is the conditional mean $\expcs{}{S|V=v,W=w}$, and the resulting \ac{mmse} is the conditional variance $\var{}{S|V,W}$. In subsequent examples, we will also encounter the case $V = S\oplus W$ with binary independent sources $S$ and $W$, and the Hamming distance as distortion measure, the next lemma gives the corresponding minimum distortion based on the observation $V$.
\begin{lemma}\label{lemma:bernoptest}
    Let $V = S\oplus W$ with $S\sim \Bern(p_1)$, $W\sim \Bern(p_2)$ independent of each other and $0\le p_1, p_2\le \frac{1}{2}$. Given the distortion function $d(s,\hat{s}) = d_H(s,\hat{s}) = \mathbbm{1}\{s = \hat{s}\}$ and arbitrary estimator $h(v)$ based on the observation $V=v$, the minimum expected distortion $\expcs{}{d(S, h(V))}$ is  
    \begin{equation}
        \min_{h: \calV \to \calS}  \expcs{}{d(S, h(V))} = \min \{p_1, p_2\}.
    \end{equation}
\end{lemma}
\begin{proof}
    Let $\bar{p} = 1-p$, the joint probability distribution of $P_{SV}$ is obtained as
    \begin{equation}
    \begin{split}
        &P_{SV}(0,0) = \bar{p}_1 \bar{p}_2, \ P_{SV}(0,1) = \bar{p}_1 p_2,\\
        &P_{SV}(1,0) = p_1 p_2, \ P_{SV}(1,1) = p_1 \bar{p}_2.
    \end{split}
    \end{equation}
    From Lemma~\ref{lemma:optest}, we have the optimal estimator 
    \begin{equation*}
    \begin{split}
        h^*(0) &= \argmin_{\hat{s}\in\{0,1\}} P_{S|V}(0|0) d_H(0,\hat{s}) + P_{S|V}(1|0) d_H(1,\hat{s})\\
        &= \argmin_{\hat{s}\in\{0,1\}} P_{SV}(0,0) d_H(0,\hat{s}) + P_{SV}(1,0) d_H(1,\hat{s})\\
        &= \argmin_{\hat{s}\in\{0,1\}} \bar{p}_1 \bar{p}_2 d_H(0,\hat{s}) + p_1 p_2  d_H(1,\hat{s})\\
        &= 0,
    \end{split}
    \end{equation*}
    due to $0\le p_1, p_2\le \frac{1}{2}$, and
    \begin{equation*}
    \begin{split}
        h^*(1) &= \argmin_{\hat{s}\in\{0,1\}} P_{SV}(0,1) d_H(0,\hat{s}) + P_{SV}(1,1) d_H(1,\hat{s})\\
        &= \argmin_{\hat{s}\in\{0,1\}} \bar{p}_1 p_2 d_H(0,\hat{s}) + p_1\bar{p}_2  d_H(1,\hat{s})\\
        &= \begin{cases}
            0, &  p_1\bar{p}_2 \le \bar{p}_1 p_2 \\
            1, & p_1\bar{p}_2 > \bar{p}_1 p_2
        \end{cases}.
    \end{split}
    \end{equation*}
    The resulting distortion is thus given by
    \begin{equation*}
    \begin{split}
        &\hspace{-6mm}\expcs{}{d(S, h^*(V))}\\
        =& P_{SV}(0,0) d(0, h^*(0)) + P_{SV}(0,1) d(0, h^*(1))\\
        &+ P_{SV}(1,0) d(1, h^*(0)) + P_{SV}(1,1) d(1, h^*(1))\\
        % =& P_{SV}(1,0) + \min \{P_{SV}(1,0), P_{SV}(1,1)\}\\
        =& p_1 p_2 + \min \{p_1\bar{p}_2, \bar{p}_1 p_2\}\\
        =& \min \{p_1, p_2 \}.
    \end{split}
    \end{equation*}
\end{proof}

\section{Point-to-Point Channel}\label{sec:p2p}
\subsection{Channel Model}

The point-to-point \ac{jsmc} model is illustrated in Fig.~\ref{fig:p2pmodel}, where an encoder tries to communicate reliably with a decoder over a \ac{sddmc} $P_{Y|XS}$ and simultaneously assists the decoder in estimating the channel state $S$. More specifically, the encoder encodes a message
% \footnote{Without loss of generality, we assume $2^{nR}$ is an integer.} 
$M\in \calM=[2^{nR}]$ to a $n$-sequence $X^n$, whose elements are from the finite input alphabet $\calX$, with the help of \ac{sit} $S_T$ and the channel feedback if present. At each time step $i$, the channel state $S_i$ follows the distribution $P_S$, and the \ac{sit} is assumed to be available strictly causally. Unlike previous works, where it is assumed that $S_T=\varnothing$\cite{zhang2011joint} or $S_T=S$\cite{choudhuri2013causal}, we consider a generalized \ac{sit} generated from $P_{S_T|S}$. Upon receiving $Y^n$, the decoder jointly decodes the message $\hM$ and estimates the channel state $\hS^n$. 
It is assumed that $Y$ is statistically independent of $S_T$ conditioned on $(X,S)$. 
To analyze the impact of present and absent channel feedback under the same framework, we denote the feedback as
a deterministic function of the channel output $Z=\phi(Y) \in \calZ$, which includes but is not restricted to the cases $Z = Y$ and $Z=\varnothing$.
% It is also noted that there is no need to distinguish between (strictly) causal and non-causal \ac{sir} because the decoder can always wait until the end of reception to perform the joint task.

The model in Fig.~\ref{fig:p2pmodel} consists of the following components:
\begin{enumerate}
    \item An encoder $f_e^n = (f_{e,1}, f_{e,2}, ..., f_{e,n})$ with $f_{e,i}: \calM \times \calS_T^{i-1} \times \calZ^{i-1} \rightarrow \calX, \forall i\in[n]$;
    \item A message decoder $f_d: \calY^n \rightarrow \calM$;
    \item A state estimator $h^n$ with $h_i : \calY^n \rightarrow \hat{\calS}, \forall i\in[n]$.
\end{enumerate}

\begin{figure}[h]
    \centering
    % \begin{tikzpicture}[node distance=3em and 4em]
    %   \node[] (in) {};
    %   \node[block, right=of in] (enc) {Encoder};
    %   \node[block, right=of enc] (channel) {$P_{Y|XS}$};
    %   \node[block, above=of channel](state){$P_{SS_T}$};
    %   \node[block, right=of channel] (dec) {Decoder};
    %   \node[right=of dec] (out) {};
    %   \node[coordinate, below=of channel] (p){};
    
    %   \draw[->] (in) -- node[above] {$M$} (enc);
    %   \draw[->] (enc) -- node[above] {$X_i$} (channel);
    %   \draw[->] (state) --node [left] {$S_i$} (channel);
    %   \draw[->] (state) -| node[anchor=south]{$S_{T}^{i-1}$} (enc);
    %   \draw[->] (channel) -- node[above] {$Y_i$} (dec);
    %   \draw[-] (channel) -- (p);
    %   \draw[->] (p) -| node[anchor=south west] {$Z^{i-1}$}(enc);
    %   \draw[->] (dec) -- node[above] {$\hM, \hS^n$} (out);
    % \end{tikzpicture}
    \includegraphics[scale=0.9]{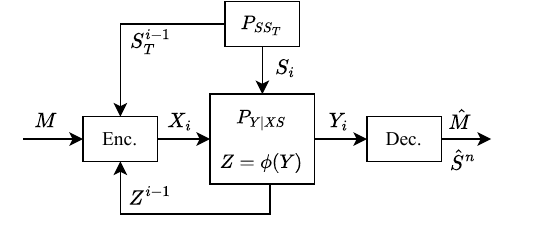}
    \caption{Channel model for point-to-point \ac{jsmc} system.}
    \label{fig:p2pmodel}
\end{figure}

The message decoding error probability is given by
\begin{equation}
    P_e^{(n)} \triangleq \frac{1}{|\calM|}\sum_{m\in \calM}\mathrm{Pr}\{f_d(Y^n) \neq M | M=m\}.
\end{equation}
We further define a distortion function to measure the estimation error as $d: \calS \times \hat{\calS} \rightarrow [0, \infty)$, such that the expected distortion is
\begin{equation}
    D^{(n)}\triangleq \expcs{}{d^n(S^n, \hS^n)} = \frac{1}{n}\sum_{i=1}^n\expcs{}{d(S_i, h_i(Y^n))}.
\end{equation}

\begin{remark}
    If \ac{sir} $S_R$ is available following $P_{S_R|S}$, $S_R$ and $Y$ can be viewed as joint channel outputs according to $P_{YS_R|XS}=P_{Y|XS}P_{S_R|S}$. 
    Therefore, \ac{sir} and normal channel outputs will not be distinguished in the following.
\end{remark}

\subsection{Main Results}

\begin{definition}\cite{zhang2011joint,ahmadipour2022information}\label{def:cdf}
    A rate-distortion pair $(R, D)$ is achievable if there exists a sequence of $(2^{nR}, n)$ codes such that 
    \begin{equation}
    \begin{split}
        \lim_{n\rightarrow \infty} P_e^{(n)} = 0,\ \limsup_{n\rightarrow \infty} D^{(n)} \le D.
    \end{split}
    \end{equation}
    The \ac{cd} function of the point-to-point model in Fig.~\ref{fig:p2pmodel}, denoted as $C(D)$, is the supremum of $R$ such that $(R, D)$ is achievable.
\end{definition}

\begin{remark}
Another criterion of interest in existing works\cite{ahmadipour2022information, zhang2011joint} is the input cost, which imposes additional constraints, such as the average input power, for the transmit signals $X$. Under this setting, Definition~\ref{def:cdf} is extended to the capacity-distortion-cost function. Because the additional input cost does not change the following analysis by adding the constraint on $P_X$ into the random variable set $\calP_D$ defined in \eqref{eq:pd}, it is omitted in this manuscript for the sake of simplicity.
\end{remark}

In \cite{choudhuri2013causal,bross2017rate}, the \ac{cd} functions for the cases $S_T = S$ and $Z = Y$ as well as $Z=\varnothing$ are derived. \cite{choudhuri2013causal} proves the achievability for the case $Z=\varnothing$, whose coding scheme can be extended to our setup but requires the Wyner-Ziv random binning technique and the decoding of message and state description is performed sequentially, as discussed later. 
In\cite{bross2017rate}, the coding scheme with perfect feedback does not need the random binning method, and the state description is simultaneously decoded with the message. However, it relies on the perfect knowledge of feedback at the encoder and becomes suboptimal for the general case $Z=\phi(Y)$. Here, we propose a new coding scheme in which the message and state description are decoded simultaneously in the backward direction, requiring neither the random binning technique nor perfect feedback at the encoder. 
% We will show its optimality for the point-to-point case and benefits brought to multi-user channels.

Let $\calP_D$ be the set of all random variables $(X, V)\in \calX\times \calV$ and functions $h: \calX \times \calV \times \calY \to \hat{\calS}$
such that $V-(X,S_T,Z)-Y$ forms a Markov chain, i.e.,
\begin{equation}\label{eq:pd}
\begin{split}
    \calP_D &\triangleq \{(X,V,h) | P_{SS_TXYZV}(s, s_T, x, y,z, v)=P_{S_TS}(s_T,s)\\
    &\cdot P_{X}(x) P_{Y|XS}(y|x,s)\mathbbm{1}\{z=\phi(y)\}P_{V|XS_TZ}(v|x,s_T,z);\\
    &\expcs{}{d(S, h(X,V,Y))}\le D\}.
\end{split}
\end{equation}
% where $\mathbbm{1}\{\cdot \}$ is the indicator function.

\begin{theorem}\label{thm:cdp2p}
    The \ac{cd} function $C(D)$ for the channel model in Fig.~\ref{fig:p2pmodel} is given by
    \begin{equation}\label{eq:p2prdfunc}
        C(D) =  \max_{(X,V,h)\in \calP_D} I(X;Y) - I(V; S_T|X,Y).
    \end{equation}
\end{theorem}
\begin{proof}
    The proof of converse is conducted similarly as in \cite{bross2017rate} with slight modifications on \ac{sit} and feedback. The coding scheme in \cite{choudhuri2013causal} can also be extended to this case. We provide a sketch of the proposed coding scheme. The details can be found in 
    \if \doublecol1 
        \cite[Appendix~A]{li2025simultaneousdecodingapproachjoint}.
    \else 
        Appendix~\ref{app:p2p-thm}. 
    \fi

    In each transmission block $b$, the encoder generates $2^{n(R+R_s)}$ sequences $x^n(b) = x^n(m_b, l_{b-1})$ from $P^n_X$ with $m_b\in[2^{nR}]$ the message in block $b$ and $l_{b-1}\in [2^{nR_s}]$ the state description index in block $b-1$. For each $x^n(b)$, it also generates $2^{nR_s}$ sequences $v^n(b) = v^n(l_b | m_b, l_{b-1})$ from $P^n_{V|X}$ to compress the \ac{sit} and feedback $(s_T^n(b), z^n(b))$ in block $b$ by finding a jointly typical sequence from $\typset{}{P_{S_TZXV}}$. This requires $R_s > I(V; S_T, Z |X)$. In the last block $B+1$, only $x^n(1, l_B)$ is sent without any message.

    The decoder starts from the last block to first recover $l_B$. Subsequently, it decodes $(m_b, l_{b-1})$ jointly from $b=B$ to $b=1$. This requires $R+R_s < I(X,V;Y)$. The channel state is recovered by applying $h(x_i(b), v_i(b), y_i(b))$ for all $i\in [n]$ and $b\in [B]$. Consequently, the rate 
    \begin{equation}
    \begin{split}
        R &< I(X, V; Y) - I(V; S_T, Z | X)\\
       &= I(X; Y) + I(V; Y |X) -  I(V; S_T, Z | X) \\
       &= I(X; Y) - I(V; S_T | X,Y)
    \end{split}
    \end{equation}
    and the distortion $D$ are achievable as $n, B \to \infty$.
\end{proof}

In the proposed coding scheme, only one index is used to compress the \ac{sit} and feedback from the previous block. As a result, it can be decoded simultaneously with the message. In contrast, \cite{choudhuri2013causal} employs the random binning technique, designing two indices for the state compression, and only the bin index is sent to the decoder. Therefore, in each block, the decoder first recovers the message with the bin index and then the other fine index. Although both coding schemes lead to the same result for the point-to-point channel, it will be seen that the simultaneous decoding approach can provide better regions for degraded broadcast and multiple access channels.

\begin{prop}\label{prop:cdcausalconcave}
$C(D)$ has the following properties:
    \begin{enumerate}
        \item $C(D)$ is a non-decreasing and concave function in $D$;
        \item The maximization in \eqref{eq:p2prdfunc} is achieved by the optimal estimator
        \begin{equation}\label{eq:p2poptest}
        h^*(x,v,y) = \argmin_{\hat{s}\in \hat{\calS}} \sum_{s\in\calS}P_{S|XVY}(s|x,v,y)d(s,\hat{s});
        \end{equation}
        \item The minimum achievable distortion is given by
        \begin{equation}
        \begin{split}
            D_{\min} \triangleq& \min_{P_X, P_{V|XS_TZ}}\expcs{}{d(S,h^*(X,V,Y))},\\
            &\mathrm{s.\ t.}\  I(X;Y) - I(V; S_T|X,Y)\ge 0.
        \end{split}
        \end{equation}
        % \item To evaluate $C(D)$, the cardinality of $\calV$ may be restricted to $|\calV| \le |\calS_T|+1$;
        \item If $S_T$ is statistically independent of $S$, it suffices to set $V=\varnothing$ and 
        \begin{equation}
            C(D) =  \max_{(X,\varnothing,h)\in \calP_D} I(X;Y).
        \end{equation}
    \end{enumerate}
\end{prop}
\begin{proof}
    See
    \if \doublecol1
        \cite[Appendix~B]{li2025simultaneousdecodingapproachjoint}.
    \else
        Appendix~\ref{app:p2p-prop}.
    \fi
\end{proof}

\begin{remark}
    If the \ac{sit} is causally available, i.e., $S_T^i$ is present at the encoder at time step $i$, the corresponding \ac{cd} function $C_{\C}(D)$ can be derived by applying the Shannon strategy\cite{bross2017rate, choudhuri2013causal}. This introduces an additional random variable $U$ and a Shannon strategy function $f_e$ with $X=f_e(U, S_T)$ such that $U-(X, S_T) - Y$ and $V-(U,S_T,Z)-Y$ form two Markov chains. In this case, $C_{\C}(D)$ is given by
    \begin{equation}
        C_{\C}(D) \triangleq \max_{\substack{P_U, f_e, P_{V|US_TZ}:\\ \expcs{}{d(S, h(U,V,Y))} \le D}} I(U; Y) - I(V; S_T| U,Y).
    \end{equation}
    % The achievability proof can also be done using either the sequential method in \cite{choudhuri2013causal} or the proposed coding scheme. 
    The Shannon strategy can also be applied to the multi-user systems in subsequent sections and extend the results to the case of causal \ac{sit}.
\end{remark}

\subsection{Examples}

\subsubsection{Radar Systems}

In radar systems, a transmitter sends pre-designed signals toward a target, and a receiver, knowing the transmit signals, estimates the target state via the reflected signals. In the monostatic setting, the transmitter and receiver are co-located, and the received signals are referred to as echo signals, also present at the transmitter. In bistatic systems, transmitters and receivers are separated.
    Existing works analyze both radar types with different models\cite{ahmadipour2022information, chang2023rate}, but we recognize their inherent equivalence by modeling them as point-to-point channels and treating the echo signals in monostatic radar as both channel output and feedback. The known transmitted signals at the receiver can be treated as an additional channel output. Thus, by interpreting $S_T=\varnothing$, $Y = (Y', X)$ with $Y'$ the reflected signal, $Z= Y'$ for the monostatic system and $Z=\varnothing$ for the bistatic system, one can set $V=\varnothing$ according to Proposition~\ref{prop:cdcausalconcave}, and the minimum estimation distortion is obtained as
    \begin{equation}\label{eq:radardistortion}
    \begin{split}
        \min_{P_X}\expcs{}{d(S,h^*(X,Y))}&=\min_{P_X}\expc{}{\expc{}{d(S,h^*(X,Y'))} | X}\\
        &= \min_{x\in \calX} \expc{}{d(S,h^*(x,Y'))}.
    \end{split}
    \end{equation}
    This implies that the minimum distortion can be attained by a deterministic transmit signal, and the echo signal, which is used to design the radar signals in real time, cannot improve the system's performance due to the assumption of a memoryless system. 
    Notably, the minimum achieved distortion may not benefit from the receiver's knowledge of $X$ if the receiver has decoding capability. 

Consider a bistatic radar system discussed above but without the knowledge of transmitted signals at the receiver. The channel is modeled as a Rayleigh fading channel $Y=SX+W$ 
with $X \in \{0, -1, +1\}$, $S$ and $W$ following $\calN(0,1)$ and mutually independent. The distortion $d$ is chosen as the squared error. From \eqref{eq:radardistortion}, it can be shown that $X\in \{-1, +1\}$ both attain the minimum distortion due to the symmetry, and thus, one can still convey data in this case. The maximum data rate is achieved by a uniform distribution $P_X(-1) = P_X(1) = \frac{1}{2}$.

\subsubsection{Quadratic-Gaussian Case}\label{sec:p2pcqg}

We consider the additive Gaussian channel
\begin{equation}\label{eq:qgchannel}
    Y= X+S+W
\end{equation}
with $S\sim \calN(0, Q)$, $W\sim \calN(0, N)$ and the squared error function $d(s,\hat{s})=(s-\hat{s})^2$. The input signal power is constrained by $P$, i.e., $\expcs{}{|X|^2} \le P$. The encoder is assumed to observe the noisy version of $S$
\begin{equation}\label{eq:qgsit}
    S_T = S + W_T
\end{equation}
with $W_T \sim \calN(0, N_T)$. 
\begin{prop}\label{prop:cqg}
    The \ac{cd} function of the additive Gaussian channel \eqref{eq:qgchannel} with generalized \ac{sit} \eqref{eq:qgsit} and feedback $Z=\phi(Y)$ for any deterministic function $\phi$ is given by
    \begin{equation}\label{eq:cqg}
    \begin{split}
        &C_{\mathrm{QG}}(D, P, Q, N, N_T) \triangleq \frac{1}{2}\left(\log\br{1 + \frac{P}{Q+N}} - \right. \\
        &\left.\log^+\frac{Q^2N^2}{(Q+N)(DQN + DQN_T + DNN_T - QNN_T)}\right)
    \end{split}
    \end{equation}
    for $D\ge \frac{Q^2N^2 + QNN_T(P+Q+N)}{(P+Q+N)(QN+QN_T+NN_T)}$, where $\log^+(x) = \max(0, \log (x))$.
\end{prop}

Note that this result reduces to the one in\cite{choudhuri2013causal, bross2017rate} for $N_T = 0$, i.e., $S_T = S$. As $N_T \to \infty$, $C_{\mathrm{QG}}(D, P, Q, N, N_T)$ tends to $\frac{1}{2}\log\br{1 + \frac{P}{Q+N}}$ for $D \ge \frac{QN}{Q+N}$, where $\frac{QN}{Q+N}$ corresponds to the \ac{mmse} of $S$ given observation $S+W$, which is the result in \cite{zhang2011joint}. 
The optimal random variables achieving $C_{\mathrm{QG}}(D, P, Q, N, N_T)$ are $X \sim \calN(0, P)$, $V = S_T + E$ where $E\sim \calN(0, d^2)$ is independent of others and 
\begin{equation}\label{eq:cqgd2}
    d^2 = \frac{DQN + DQN_T + DNN_T - QNN_T}{QN - DQ - DN}.
\end{equation}
The optimal estimator is the \ac{mmse} estimator 
\begin{equation}
\begin{split}
    h_{\mathrm{QG}}^*(v,x,y) &\triangleq \expcs{}{S|V=v, X=x, Y=y}\\
    &= \frac{QNv + (QN_T+d^2Q)(y-x)}{QN + QN_T + NN_T + d^2(Q+N)}.
\end{split}
\end{equation}
An example with $P = 5$, $Q=N=1$ and $N_T \in \{0, 0.3, 1\}$ is illustrated in Fig.~\ref{fig:p2pgaussian}.

\begin{figure}
    \centering
    \includegraphics[scale=0.8]{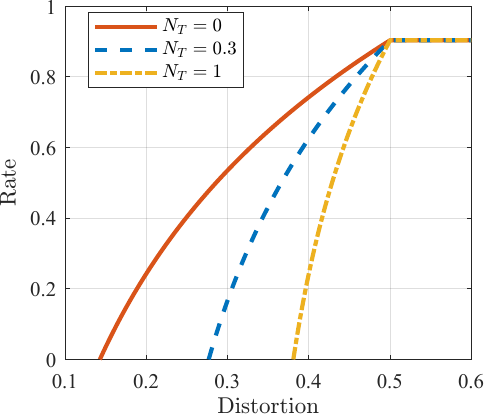}
    \caption{\ac{cd} function for the quadratic-Gaussian case with different correlation levels between $S$ and $S_T$.}
    \label{fig:p2pgaussian}
\end{figure}

\begin{proof}
    Although Theorem~\ref{thm:cdp2p} is only proved for finite alphabet, the extension to Gaussian channels can be conducted using techniques in\cite{el2011network,wyner1978rate}. The proof of achievability is straightforward, and we shall only show the converse direction for the perfect feedback case.

    First, we will show the dependency of $V$ and $h$ on $X$ can be removed. Let $S_W = S+W$ and note that the dependency of $V$ on $(S_T, X, Y)$ and that of $h$ on $(V, X, Y)$ can be replaced by the dependency on $(S_T, X, S_W)$ and $(V, X, S_W)$, respectively, due to the one-to-one correspondence $S_W = Y-X$, and $I(V;S_T | X,Y) = I(V;S_T | X, S_W)$. In turn, for any fixed probability density function\footnote{The proof can also be extended to non-absolute continuous distributions, which do not possess probability density functions. But we make this assumption here for simplicity.} $p_{V|S_TXS_W}$, define a new random variable $V'$ conditioned only on $(S_T, S_W)$ with
    \begin{equation*}
        p_{V'|S_TS_W}(v|s_T, s_W) = \int_{\calX} p_{V|S_TXS_W}(v|s_T,x,s_w) p_{X}(x)\mathrm{d}x
    \end{equation*}
    for any $p_{X}$. Thus, $(V', S_T, S_W)$ is independent of $X$. It can be shown that $I(V';S_T | X,S_W) \le I(V;S_T | X,S_W)$ due to the 
    % convexity of mutual information and $p_{V'|S_TS_W}(v|s_T, s_W)$ is a convex combination of $p_{V|S_TXS_W}$
    log-sum inequality. The estimator may also not depend on $X$ according to Lemma~\ref{lemma:markovest} and the Markov chain $S - (V', S_W) - X$. In other words, it suffices to design $p_{V|S_TS_W}$ and the estimator $h(V, S_W)$ to optimize the \ac{cd} function. This leads to $I(V;S_T | X, S_W) = I(V;S_T |S_W)$ and thus we have
    \begin{equation}
    \begin{split}
        &\hspace{-7mm}C_{\mathrm{QG}}(D, P, Q, N, N_T) =\\
        &\max_{p_X} I(X;Y)- \min_{p_{V|S_TS_W}, h} I(V; S_T | S_W), \\
        &\mathrm{s.\ t.} \ \expcs{}{|X|^2}\le P,\ \expcs{}{(S-h(V,S_W))^2}\le D.
    \end{split}
    \end{equation}
    The first term is the capacity of additive Gaussian channel $\frac{1}{2}\log (1+\frac{P}{Q+N})$ and the second term, taking the form of noisy Wyner-Ziv \ac{rd} function of Quadratic-Gaussian case\cite{draper2002successive}, has the result
    \begin{equation*}
    \begin{split}
        &\frac{1}{2} \log^+ \frac{\var{}{S|S_W} - \var{}{S|S_TS_W}}{D - \var{}{S|S_TS_W}}\\
        = & \frac{1}{2} \log^+\frac{Q^2N^2}{(Q+N)(DQN + DQN_T + DNN_T - QNN_T)}.
    \end{split}
    \end{equation*}
    This completes the proof.

\end{proof}

\section{Degraded Broadcast Channel}\label{sec:braodcast}
\subsection{Channel Model}

A two-user \ac{sddmbc} model for \ac{jsmc} is shown in Fig.~\ref{fig:bcmodel}, where an encoder maps two private messages $M_1\in  \calM_1= [2^{nR_1}]$ and $M_2\in  \calM_2= [2^{nR_2}]$ dedicated to user 1 and 2 into $X^n$. The state-dependent channel is characterized by $P_{Y_1Y_2|XS}$. The state $S$ and \ac{sit} $S_T$ are defined in the same way as the point-to-point case, and the channel outputs $(Y_1, Y_2)$ are assumed to be independent of $S_T$ conditioned on $(X,S)$. Decoder $k\in\{1,2\}$ receives the channel output $Y_k^n$, simultaneously decodes the message $\hM_k$ and estimates the channel state $\hS_k^n$. The \ac{sddmbc} feedback $Z=\psi(Y_1, Y_2)$ is expressed as a function of $(Y_1, Y_2)$.
% with $Y_1' = \phi_1(Z_1)$ and $Y_2'=\phi_2(Z_2)$ to include the possible cases $\{(Y_1, Y_2), (Y_1,\varnothing), (\varnothing,Y_2), \varnothing\}$
The system consists of the following components:
\begin{enumerate}
    \item An encoder $g_e^n$ with $g_{e,i}: \calM_1 \times \calM_2 \times \calS_T^{i-1}\times\calZ^{i-1}\to \calX, \forall i\in [n]$;
    \item Message decoders $g_{d,k}: \calY_k^n \to \calM_k$ at Decoder $k \in \{1,2\}$;
    \item State estimators $h_k^n$ with $h_{k,i}: \calY_k^n \to \hat{\calS}_k$ at Decoder $k \in \{1,2\}, \forall i\in [n]$.
\end{enumerate}

\begin{figure}[h]
    \centering
    % \begin{tikzpicture}[node distance=4em and 4em]
    %   \def\msgsps{2.5em}
    %   \node[bigblock] (enc) {Encoder};
    %   \node[left=of enc, yshift=\msgsps] (msg0) {};
    %   \node[left=of enc] (msg1) {};
    %   \node[left=of enc, yshift=-\msgsps] (msg2) {};
    %   \node[txtbigblock, right=of enc] (channel) {$P_{Y_1Y_2|XS}$};
    %   \node[block, above=2em of channel](state){$P_{SS_T}$};
    %   \node[block, right=of channel, yshift=\msgsps] (dec1) {Decoder 1};
    %   \node[block, right=of channel, yshift=-\msgsps] (dec2) {Decoder 2};
    %   \node[right=4em of dec1] (out1) {};
    %   \node[right=4em of dec2] (out2) {};
    %   \node[coordinate, below=3em of channel] (p){};
    
    %   \draw[->] (msg0) -- node[above] {$M_0$} ([yshift=\msgsps]enc.west);
    %   \draw[->] (msg1) -- node[above] {$M_1$} (enc);
    %   \draw[->] (msg2) -- node[above] {$M_2$} ([yshift=-\msgsps]enc.west);
    %   \draw[->] (enc) -- node[above] {$X_i$} (channel);
    %   \draw[->] (state) --node [left] {$S_i$} (channel);
    %   \draw[->] (state) -| node[anchor=south]{$S_{T}^{i-1}$} (enc);
    %   \draw[->] ([yshift=\msgsps]channel.east) -- node[above] {$Y_{1,i}$} (dec1);
    %   \draw[->] ([yshift=-\msgsps]channel.east) -- node[above] {$Y_{2,i}$} (dec2);
    %   \draw[-] (channel) -- (p);
    %   \draw[->] (p) -| node[anchor=south west] {$Z^{i-1}$}(enc);
    %   \draw[->] (dec1) -- node[above] {$\hM_0(1)$} node[below] {$\hM_1, \hS_1^n$} (out1);
    %   \draw[->] (dec2) -- node[above] {$\hM_0(2)$} node[below] {$\hM_2, \hS_2^n$} (out2);
    % \end{tikzpicture}
    \includegraphics[scale=0.9]{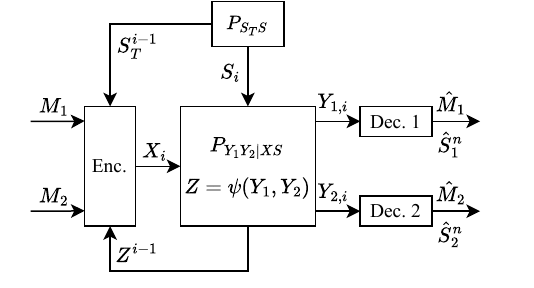}
    \caption{The degraded broadcast \ac{jsmc} model.}
    \label{fig:bcmodel}
\end{figure}

\begin{definition}\label{def:degbc}
    A state-dependent broadcast channel is said to be physically degraded if $(X, S_T) - Y_1 - Y_2$ forms a Markov chain. A state-dependent broadcast channel is said to be statistically degraded if there exists a random variable $\tilde{Y}_1$ such that $P_{\tilde{Y}_1|XS}(\tilde{y}_1 | x,s) = P_{Y_1|XS}(\tilde{y}_1 | x,s)$ and
    \begin{equation}
    \begin{split}
        &\sum_{s\in \calS} P_{SS_T}(s, s_T) P_{\tilde{Y_1}, Y_2 | XS}(y_1, y_2 | x,s) =\\
        &\hspace{2mm}P_{Y_2 | \tilde{Y_1}} (y_2 | y_1) \sum_{s\in \calS} P_{SS_T}(s, s_T) P_{\tilde{Y_1}| XS}(y_1 | x,s).
    \end{split}
    \end{equation}
    In other words, $(X, S_T) - \tilde{Y}_1 - Y_2$ forms a Markov chain for arbitrary $P_X$.
\end{definition}
It is noted that $(X,S)-Y_1-Y_2$ implies the Markov chain $(X,S_T)-Y_1-Y_2$ but not vice versa. Thus, Definition~\ref{def:degbc} is less strict than the assumptions made in \cite{gel1980coding, bross2020message}.

Let $d_k: \calS \times \hat{\calS}_k \to [0, \infty)$ be the distortion function at Decoder $k$, which is not necessarily identical for different $k$. For instance, let $S=(S_1, S_2)$, Decoder 1 only concerns about $S_1$ and Decoder 2 only about $S_2$. This case can be handled by properly designing $d_k$ for both decoders, which only evaluates the part of interest $S_k$. In a similar way as the point-to-point channel, one can define the corresponding message decoding error probability $P_e^{(n)}$ and averaged distortion $D^{(n)}_k$ as well as the \ac{cd} region.
\begin{definition}
    A tuple of rates and distortions $(R_1, R_2, D_1, D_2)$ is said to be achievable if there exists a sequence of $(2^{nR_1}, 2^{nR_2}, n)$ codes such that for all $k\in\{1,2\}$
    \begin{align}
            \lim_{n\rightarrow \infty} P_e^{(n)} = 0,\ \limsup_{n\rightarrow \infty} D_k^{(n)} \le D_k.
    \end{align}
    The \ac{cd} region of the \ac{sddmbc} in Fig.~\ref{fig:bcmodel}, denoted as $\calC^{\BC}(D_1, D_2)$, is the closure of the set of $(R_1, R_2)$, such that $(R_1, R_2, D_1, D_2)$ is achievable.
\end{definition}

\subsection{Main Results}

An achievable region is first derived based on the proposed coding scheme, denoted as $\calR^{\BC}(D_1, D_2)$. In the coding scheme, the encoder adopts the superposition coding strategy and successive refinement method\cite{steinberg2004successive} for compressing the side information due to the fact that the stronger user can always decode the codewords dedicated to the weak user.
% With the proposed method, there is no need to apply the random binning technique, and the simultaneous decoding approach is performed. 
We will show that, although the Wyner-Ziv random binning and sequential decoding method can also be applied, the proposed scheme results in a larger region.

Let $\pdd^{\BC}$ be the set of all random variables $(U, X, V_1, V_2)$ and functions $(h_1, h_2)$ with $h_1: \calU\times \calX\times \calV_1 \times \calV_2 \times \calY_1 \to \hat{\calS}_1$ and $h_2: \calU\times \calV_2 \times \calY_2 \to \hat{\calS}_2$ such that:
\begin{equation}\label{eq:pddbc}
\begin{split}
    &\hspace{-3mm}\pdd^{\BC} \triangleq \left\{ (U, X, V_1, V_2, h_1, h_2)\right|\\
    &P_{SS_TUXY_1Y_2ZV_1V_2}(s,s_T,u,x,y_1,y_2,z,v_1,v_2)=\\
    &P_{SS_T}(s,s_T) P_{UX}(u,x)P_{Y_1 Y_2 | X S} (y_1, y_2 | x , s)\\
    &\cdot\mathbbm{1}\{z=\psi(y_1, y_2))\} P_{V_1V_2|UXS_TZ}(v_1,v_2|u,x,s_T,z); \\
    &\expcs{}{d_1(S, h_1(U, X, V_1, V_2, Y_1))} \le D_1,\\
    &\left.\expcs{}{d_2(S, h_2(U, V_2, Y_2))} \le D_2 \right \},
\end{split}
\end{equation}
in which we have the Markov chains $U - (X, S_T) - Y_1 - Y_2$ and $(V_1, V_2) - (U, X, S_T, Z) - (Y_1, Y_2)$.
% \begin{equation*}
% \begin{split}
%     &U - (X, S_T) - Y_1 - Y_2\\
%     &(V_1, V_2) - (U, X, S_T, Z) - (Y_1, Y_2).
% \end{split}
% \end{equation*}
Note that this does not imply $(V_1, V_2)- Y_1 - Y_2$ in general. 
The derived \ac{rd} region $\calR^{\BC}(D_1, D_2)$ contains all nonnegative rate pairs $(R_1, R_2)$ such that
\begin{equation}\label{eq:rdregion}
\begin{split}
    & R_2 <  I(U; Y_2) - I(V_2; S_T, Z, X| U, Y_2),\\
    & R_1 < I(X; V_2, Y_1 | U) - I(V_1 ; S_T, Z|U, X, V_2,  Y_1)
\end{split}
\end{equation}
% \begin{equation}\label{eq:rdregion}
% \begin{split}
%     &\calR^{\BC}(D_1, D_2) \triangleq \bigcup_{(U, X, V_1, V_2, h_1, h_2) \in \pdd^{\BC}}\\
%     &\quad \brcur{
%         \begin{array}{ll}
%                            & R_1\ge 0, R_2 \ge 0,\\
%         (R_1, R_2):   & R_2 <  I(U; Y_2) - I(V_2; S_T, Z, X| U, Y_2),\\
%                            & R_1 < I(X; V_2, Y_1 | U) - I(V_1 ; S_T, Z|U, X, V_2,  Y_1)\\
%         \end{array}
%         }.
% \end{split}
% \end{equation}
for $(U, X, V_1, V_2, h_1, h_2) \in \pdd^{\BC}$.

$\calR^{\BC}(D_1, D_2)$ can be understood as follows. If there is no state estimation task at both decoders, the region reduces to the well-known capacity region for the degraded broadcast channel $\{R_2<I(U;Y_2), R_1 <I(X;Y_1|U)\}$ by setting $V_1= V_2=\varnothing$, where $U$ represents the shared information.
% , and it has been shown that strictly causal \ac{sit} and perfect feedback do not improve it\cite{gamal1978feedback}. 
Having the state estimation task, the encoder compresses the \ac{sit} $S_T$ with the knowledge of $(U, X, Z)$ into two state descriptions $(V_1, V_2)$ to the respective decoder by splitting rates from the total amount of the achievable rates. Since Decoder 1 is capable of recovering all codewords at Decoder 2, including $V_2$, it can estimate the channel state based on both state descriptions. Therefore, we characterize the state compression as a successive refinement problem, where $V_2$ is recovered by both decoders at the first stage while $V_1$ is only by Decoder 1 at the second refinement stage. The two terms $I(V_2; S_T, Z, X| U, Y_2)$ and $I(V_1 ; S_T, Z|U, X, V_2,  Y_1)$ in $\calR^{\BC}(D_1, D_2)$ thus take a similar form to the \ac{rd} region derived in\cite{steinberg2004successive} with Decoder 2 and 1 cognizant of $(U, Y_2)$ and $(U,X,Y_1)$, respectively. In addition, we note that the total rate for $R_1$ becomes $I(X; V_2,Y_1|U)$ instead of $I(X;Y_1|U)$. The reason is that, besides being used as the first stage description, $V_2$ can also be treated as another channel output to Decoder 1 due to its dependency on $(S_T,X)$, thus assisting in decoding $X$.

\begin{theorem}\label{theorem:deg-bc-scc}
    The \ac{cd} region for the degraded \ac{sddmbc} in Fig.~\ref{fig:bcmodel} satisfies
    \begin{equation}
        \calR^{\BC}(D_1, D_2) \subseteq \calC^{\BC}(D_1, D_2).
    \end{equation}
\end{theorem}
\begin{proof}
    See 
    \if \doublecol1
        \cite[Appendix~C]{li2025simultaneousdecodingapproachjoint}
    \else
        Appendix~\ref{app:deg-bc-scc}
    \fi
    for details. Here, we provide a sketch of the coding scheme.

    In each block $b$, the encoder generates $2^{n(R_2+R_{s2})}$ sequences of $u^n(b) = u^n(m_{2,b}, l_{2,b-1})$ from $P^n_U$, $2^{n(R_1+R_{s1})}$ sequences of $x^n(b) = x^n(m_{1,b}, l_{1,b-1} | m_{2,b}, l_{2,b-1})$ given each $u^n(b)$ from $P^n_{X|U}$, $2^{nR_{s2}}$ sequences $v_2^n(b) = v_2^n(l_{2,b}|m_{2,b}, l_{2,b-1})$ given each $u^n(b)$ from $P_{V_2|U}$, and $2^{nR_{s1}}$ sequences $v_1^n(b) = v_1^n(l_{1,b}|m_{2,b}, l_{2,b-1}, m_{1,b}, l_{1,b-1}, l_{2,b})$ given each $(u^n(b), x^n(b), v^n_2(b))$ from $P^n_{V_1|UXV_2}$, where $m_{1,b}\in [2^{nR_1}]$, $m_{2,b}\in [2^{nR_2}]$ are the message for Decoder 1 and 2 in block $b$, respectively. The state descriptions $l_{1,b}\in [2^{nR_{s1}}]$, $ l_{2,b}\in [2^{nR_{s2}}]$ are obtained by finding a jointly typical sequence from $\typset{}{P_{S_TZUXV_2V_1}}$, which requires
    \begin{equation}
    \begin{split}
        R_{s2} & > I(V_2; S_T, Z, X |U),\\
        R_{s1}& > I(V_1; S_T, Z|U,X,V_2).
    \end{split}
    \end{equation}
    In the last block $B+1$, only $x^n(1,l_{1,B} | 1, l_{2,B})$ is sent.

    Decoder 2 follows the same steps as the decoder of the point-to-point channel to recover $(m_{2,b}, l_{2,b-1})$ simultaneously and estimate the state, requiring 
    \begin{equation}
        R_2+R_{s2} < I(U,V_2; Y_2).
    \end{equation}
    Decoder 1 starts from the last block with decoding $(l_{1,B}, l_{2,B})$ jointly. Afterwards, it decodes $(m_{1,b}, m_{2,b}, l_{1,b-1}, l_{2,b-1})$ by looking for a jointly typical sequence from $\typset{}{P_{UV_2XV_1Y_1}}$. As proved in the appendix, the correct simultaneous decoding is guaranteed by
    \begin{equation}
        R_1 + R_{s1} < I(X; V_2, Y_1 |U) + I(V_1; Y_1 |U,X, V_2).
    \end{equation}
    The state at both decoders is estimated as $h_1(u_i(b), x_i(b), v_{1,i}(b), v_{2,i}(b), y_{1,i}(b))$ and $h_2(u_i(b), v_{2,i}(b), y_{2,i}(b))$ for $i\in [n]$ and $b\in [B]$.
    Combining the above inequalities, we obtain the results.
\end{proof}

\begin{prop}\label{prop:deg-bc-scc}
$\calR^{\BC}(D_1, D_2)$ has the following properties:
    \begin{enumerate}
        \item If $D_1'\ge D_1$ and $D_2'\ge D_2$, then $\calR^{\BC}(D_1, D_2)  \subseteq \calR^{\BC}(D_1', D_2') $;
        \item Given two regions $(R_1', R_2') \in \calR^{\BC}(D_1', D_2') $ and $(R_1'', R_2'') \in \calR^{\BC}(D_1'', D_2'') $, let $0\le\lambda \le 1$, $R_1 = \lambda R_1'+(1-\lambda) R_1''$, $R_2=\lambda R_2'+ (1-\lambda) R_2''$, and $D_1=\lambda D_1' + (1-\lambda) D_1''$, $D_2 = \lambda D_2' + (1-\lambda) D_2''$, then $(R_1, R_2) \in \calR^{\BC}(D_1, D_2) $;
        \item $\calR^{\BC}(D_1, D_2) $ is a convex set;
        \item In order to exhaust the region $\calR^{\BC}(D_1, D_2) $, $(h_1, h_2)$ are given by the respective optimal estimators
        \begin{align*}
            &h_1^*(u,x,v_1,v_2,y_1) =\\
            &\hspace{3mm}\argmin_{\hat{s}\in \hat{\calS}} \sum_{s\in \calS} P_{S|UXV_1V_2Y_1}(s|u,x,v_1,v_2,y_1)d_1(s,\hat{s}),\\
            &h_2^*(u,v_2,y_2) =\\
            &\hspace{3mm}\argmin_{\hat{s}\in \hat{\calS}} \sum_{s\in \calS} P_{S|UV_2Y_2}(s|u,v_2,y_2)d_2(s,\hat{s}).
        \end{align*}
        % are the optimal estimators for $\hS_1$ and $\hS_2$, respectively.
    %     \item In order to exhaust the region $\calR^{\BC}(D_1, D_2) $, the cardinalities of $\calV_1$ and $\calV_2$ may be restricted to \todo{check this again}
    % \begin{equation*}
    % \begin{split}
    %     &|\calV_1| \le |\calS_T| +2, \\
    %     &|\calV_2| \le |\calS_T| +2.
    % \end{split}
    % \end{equation*}
    \end{enumerate}
\end{prop}
\begin{proof}
    See
    \if \doublecol1
        \cite[Appendix~D]{li2025simultaneousdecodingapproachjoint}.
    \else
        Appendix~\ref{app:bc-prop}.
    \fi
    % Appendix~\ref{app:bc-prop}.
\end{proof}

Another achievable \ac{rd} region based on
the random binning and sequential decoding approach can also be derived, as detailed in\cite{li2024analysis}. With $R_{s1}$, $R_{s2}$ denoting the compression rates for $V_1$, $V_2$, respectively, Decoder 2 performs the same steps as the point-to-point case in\cite{choudhuri2013causal} to first decode $U$ and then $V_2$, requiring 
\begin{equation}
\begin{split}
    R_2 + R_{s2} &< I(U; Y_2),\\
    R_{s2} &> I(V_2; S_T, Z, X| U, Y_2).
\end{split}
\end{equation}
Decoder 1 should first recover $(U,X)$ to obtain the bin indices of $(V_1, V_2)$ with
\begin{equation}
    R_1 + R_{s1} < I(X; Y_1 | U).
\end{equation}
It is then able to look for both finer indices. For $V_2$, it requires
\begin{equation}
    R_{s2} > I(V_2; S_T, Z| U, X, Y_1),
\end{equation}
while for $V_1$
\begin{equation}
    R_{s1} > I(V_1; S_T, Z| U, X, V_2, Y_1).
\end{equation}
Combined together,
the achievable \ac{rd} region with the sequential decoding approach is derived as $\calR_{\SEQ}^{\BC}(D_1, D_2)$, containing all nonnegative $(R_1, R_2)$ with
\begin{equation}\label{eq:bc-seq-rd}
\begin{split}
    & R_2 <  I(U; Y_2) - \\
    &\hspace{5mm}\max\br{I(V_2; S_T, Z, X| U, Y_2), I(V_2; S_T, Z | U, X, Y_1)},\\
    & R_1 < I(X; Y_1 | U) - I(V_1 ; S_T, Z|U, X, V_2,  Y_1),\\
\end{split}
\end{equation}
% \begin{equation}\label{eq:bc-seq-rd}
% \begin{split}
%     &\calR_{\SEQ}^{\BC}(D_1, D_2) \triangleq\\
%     &\quad \bigcup_{(U, X, V_1, V_2, h_1, h_2) \in \pdd^{\BC}}\brcur{
%         \begin{array}{ll}
%                            & R_1\ge 0, R_2 \ge 0,\\
%         (R_1, R_2):        & R_2 <  I(U; Y_2) - \max\br{I(V_2; S_T, Z, X| U, Y_2), I(V_2; S_T, Z | U, X, Y_1)},\\
%                            & R_1 < I(X; Y_1 | U) - I(V_1 ; S_T, Z|U, X, V_2,  Y_1),\\
%         \end{array}
%         }.
% \end{split}
% \end{equation}
for $(U, X, V_1, V_2, h_1, h_2) \in \pdd^{\BC}$.
It follows obviously $\calR^{\BC}(D_1, D_2) \supseteq \calR_{\SEQ}^{\BC}(D_1, D_2)$. With sequential decoding, we note that $V_2$ is not guaranteed to be correctly decoded at Decoder 1 even if Decoder 2 can recover it since the Markov chain $V_2 - Y_1 - Y_2$ doesn't hold in general. The rate $R_{s2}$ is thus constrained twice. Additionally, because $V_2$ is decoded after $X$, it can not help to decode $X$ and the total rate at Decoder 1 in $\calR_{\SEQ}^{\BC}(D_1, D_2)$ is still $I(X;Y_1|U)$. The comparison between both derived regions highlights the benefits brought by the proposed coding scheme.
% Note that because the decoders first only recover $U$ and $X$ as in the usual degraded BC coding scheme and then decode the state descriptions $(V_1, V_2)$, the first term in $R_1$ of both regions differs from each other. Moreover, the rate for the state description $V_2$ is constrained twice in $\calR_{\SEQ}^{\BC}(D_1, D_2)$ since the Markov chain $V_2 - Y_1 - Y_2$ doesn't hold in general, resulting in further performance loss, which, in contrast, brings no problem using the proposed scheme. It is obvious that $\calR^{\BC}(D_1, D_2) \supseteq \calR_{\SEQ}^{\BC}(D_1, D_2)$. The proof of achievability for $\calR_{\SEQ}^{\BC}(D_1, D_2)$ is straightforward and thus omitted here.

An outer bound of $\calC(D_1, D_2)$ is then obtained. Define $\calO^{\BC}(D_1, D_2)$ as a set of nonnegative $(R_1, R_2)$ with
\begin{equation}
\begin{split}
    & R_2 <  I(U; Y_2) - I(V_2; Y_1| U, Y_2),\\
    & R_1 < I(X; Y_1 | U) - I(V_1 ; S_T|U, X, V_2,  Y_1),\\
    & R_1 + R_2 < I(X; Y_1) - I(S_T; V_1,V_2|U, X, Y_1, Y_2)
\end{split}
\end{equation}
% \begin{equation}\label{eq:cdregionouter}
% \begin{split}
%     &\calO^{\BC}(D_1, D_2) \triangleq\\
%     &\quad \bigcup_{(U, X, V_1, V_2, X, h_1, h_2) \in \pdd}\brcur{
%         \begin{array}{ll}
%                            & R_1\ge 0, R_2 \ge 0,\\
%         (R_1, R_2):        & R_2 \le  I(U; Y_2) - I(V_2; Y_1| U, Y_2),\\
%                            & R_1 \le I(X; Y_1 | U) - I(V_1 ; S_T|U, X, V_2,  Y_1),\\
%                            & R_1 + R_2 \le I(X; Y_1) - I(S_T; V_1,V_2|U, X, Y_1, Y_2)
%         \end{array}
%         }.
% \end{split}
% \end{equation}
for $(U, X, V_1, V_2, X, h_1, h_2) \in \pdd^{\BC}$.

\begin{theorem}\label{theorem:deg-bc-outer}
    The \ac{cd} region for the degraded \ac{sddmbc} in Fig.~\ref{fig:bcmodel} satisfies
    \begin{equation}
        \calO^{\BC}(D_1, D_2) \supseteq \calC^{\BC}(D_1, D_2).
    \end{equation}
\end{theorem}
\begin{proof}
    See
    \if \doublecol1
        \cite[Appendix~E]{li2025simultaneousdecodingapproachjoint}.
    \else
        Appendix~\ref{app:deg-bc-outer}.
    \fi
    % Appendix~\ref{app:deg-bc-outer}.
\end{proof}

\begin{remark}
    If Decoder 2 has no state estimation task, i.e., $D_2 = \infty$, and $Z$ is independent of $Y_2$, the corresponding inner bounds $\calR^{\BC}(D_1, \infty)$, $\calR_{\SEQ}^{\BC}(D_1, \infty)$, and the outer bound $\calO^{\BC}(D_1, \infty)$ coincide with each other for all feasible $D_1$ by setting $V_2 = \varnothing$. This is an extension of the setting in \cite{bross2020message}.
\end{remark}

\subsection{Examples}

\subsubsection{Monostatic-downlink ISAC}

One important \ac{isac} scenario is the monostatic-downlink system\cite{liu2022integrated, ahmadipour2022information}, where a \ac{bs} sends a signal to a \ac{ue} via the downlink channel and simultaneously estimates the channel state from the echo signal in a monostatic manner. As discussed in point-to-point channel, the monostatic sensing link can be viewed as a separated transceiver system, so the overall system is modeled as a \ac{sddmbc}, in which the encoder is the \ac{bs} transmitter, and the two decoders are \ac{bs} receiver and the \ac{ue}, respectively. Following the same setup in \cite{ahmadipour2022information}, let $S_T = \varnothing$, $Y_1 = (Y'_1, X)$ with $Y'_1$ the target echo signal, $Y_2 = (Y_2', S)$ with $Y_2'$ the normal received downlink signal, $Z = Y_1'$, it's obvious $X - Y_1 - Y_2$ forms a Markov chain for arbitrary channel such that the constructed \ac{sddmbc} is degraded. By setting $U = X$, $V_1 = V_2 = \varnothing$, we recover the results in \cite{ahmadipour2022information}. Namely, the relationship between downlink capacity and sensing distortion is given by
    \begin{equation}
    \begin{split}
        C_{\mathrm{MD}}(D) &\triangleq \max_{P_X: \expcs{}{d_1(S, h_1^*(X, Y_1'))} \le D} I(X; Y_2', S)\\
        &= \max_{P_X: \expcs{}{d_1(S, h_1^*(X, Y_1'))} \le D} I(X; Y_2'| S).
    \end{split}
    \end{equation}

\subsubsection{Quadratic-Gaussian Case}

We consider the degraded additive Gaussian broadcast channel
\begin{equation}
\begin{split}
    Y_1 &= X + S + W_1,\\
    Y_2 &= Y_1 + W_2 = X+S+W_1+W_2
\end{split}
\end{equation}
with $S\sim \calN(0,Q)$, $W_1 \sim \calN(0,N_1)$, $W_2 \sim \calN(0,N_2)$ and the distortion functions are both $d(s,\hat{s}) = (s-\hat{s})^2$. The input power is constrained by $\expcs{}{|X|^2} \le P$, and the encoder observes $S_T = S+W_T$ with $W_T \sim \calN(0,N_T)$. 
% Five types of feedback are assumed: $Z = \varnothing$, $Z = Y_1$, $Z = Y_2$, $Z = Y_1+Y_2$, $Z = (Y_1,Y_2)$. Following the coding schemes for degraded Gaussian broadcast channel\cite{el2011network} and $C_{\mathrm{QG}}(D)$ in Section~\ref{sec:p2pcqg}, 
We first investigate the behavior of the derived regions with the following choice of $(U, X, V_1, V_2)$:
\begin{equation}
\begin{split}
    &X = U+V, \\
    &U\sim \calN(0, \alpha P), \ V\sim \calN(0, (1-\alpha)P),\ 0\le \alpha \le 1,\\
    &V_2 = S_T + E_2, \ E_2\sim \calN(0, d_2^2),\\
    &V_1 = S_T + E_1, \ E_1\sim \calN(0,d_1^2),
\end{split}
\end{equation}
% \begin{equation}
% \begin{split}
%     &X = U+V, \ U\sim \calN(0, \alpha P), \ V\sim \calN(0, (1-\alpha)P),\ 0\le \alpha \le 1\\
%     &V_2 = S_T + \lambda_1 X + \lambda_2 Y_1 + \lambda_3 Y_2 + E_2, \ E_2\sim \calN(0, d_2^2),\\
%     &V_1 = S_T + \gamma_1 Y_1 + \gamma_2 Y_2 + E_1, \ E_1\sim \calN(0,d_1^2),
% \end{split}
% \end{equation}
and $U, V, E_1, E_2$ are independent of each other as well as of $S, W_1, W_2, W_T$. Note that the effective additive noise at Decoder~1 is $V+N_1+N_2$ and thus the rate of $R_2 < C_{\mathrm{QG}}(D_2, \alpha P, Q, (1-\alpha)P+ N_1+N_2, N_T)$ is achieved using either simultaneous or sequential decoding approach. Fixing $D_2$ and varying $\alpha$ and $d_1^2$, with $d_2^2$ calculated for each $\alpha$, $D_2$ based on \eqref{eq:cqgd2}, we can obtain the respective achieved \ac{rd} regions from $\calR^{\BC}(D_1, D_2)$, $\calR_{\SEQ}^{\BC}(D_1, D_2)$, $\calO^{\BC}(D_1, D_2)$.

Let $P=5$, $Q=N_1=N_2=1$, $N_T=0.3$, and $D_2=0.5$, the achieved regions of $(R_1, R_2, D_1)$ of the two inner bounds are illustrated in Fig.~\ref{fig:bcg_simseq}. As mentioned above, the same constraint results in the same range of $R_2$. On the other hand, the maximum values of $D_1$ are also the same in both regions because their respective \ac{mmse} at Decoder 1 are equal for the same $d_1^2$ and $d_2^2$. Due to the benefits of simultaneous decoding, an additional rate of $I(X; V_2|U, Y_1)$ is gained at Decoder 1, which is reflected by the maximum achieved $R_1$, annotated on both figures. In turn, a lower distortion $D_1$ can thus be reached since more rates can be split for state compression.

\begin{figure}
    % \centering
    \begin{subfigure}{\columnwidth}
    \centering
        \includegraphics[scale=0.8]{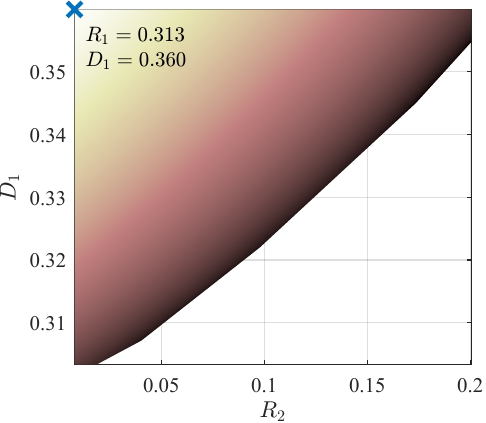}
        \caption{Simultaneous Decoding Method}
    \end{subfigure}
    \vfill
    \vspace{2mm}
    \begin{subfigure}{\columnwidth}
    \centering
        \includegraphics[scale=0.8]{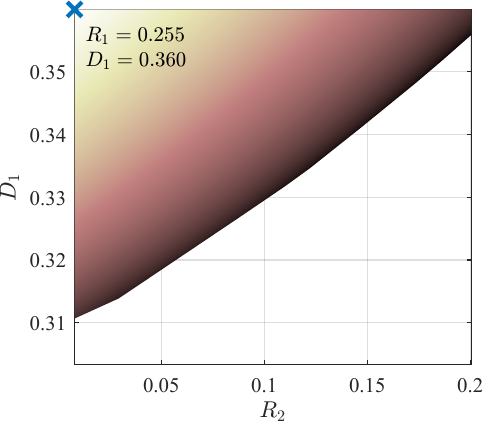}
        \caption{Sequential Decoding Method}
    \end{subfigure}
    \caption{Comparison of $(R_1, R_2, D_1)$ achieved by simultaneous and sequential decoding methods while fixing $D_2$. The black region corresponds to $R_1=0$, and lighter areas indicate higher $R_1$.}
    \label{fig:bcg_simseq}
\end{figure}

When $D_2 = \infty$, corresponding to $V_2 = \varnothing$, the resulting three regions are shown to coincide with each other and given as
\begin{equation}
\begin{split}
    R_2 &< \frac{1}{2} \log \br{1 + \frac{\alpha P}{Q+(1-\alpha)P+N_1+N_2}},\\
    R_1 &< C_{\mathrm{QG}}(D_1, (1-\alpha) P, Q, N_1, N_T).
\end{split}
\end{equation}
It turns out that such region is optimal if $Z$ is independent of $Y_2$, which can be proved similarly to Proposition~\ref{prop:cqg} by noting that $V_1$ and $h_1$ need not depend on $(U,X)$, thus the problems of finding the optimal $(U,X)$ and $(V_1, h_1)$ are decoupled, where the former follows the conventional steps in degraded Gaussian broadcast channel\cite{el2011network} and the latter results from\cite{draper2002successive}. For the case of perfect \ac{sit}, i.e., $N_T = 0$, such region reduces to the one found in\cite[Proposition~3]{bross2020message}.

\subsubsection{Improving Distortion by Compressing Codewords} 

We next show that, unlike the point-to-point channel, the achieved distortion may be improved if one adds the dependency on codewords $(U,X)$ to the state descriptions $(V_1, V_2)$ even when no \ac{sit} and feedback are present. Consider a degraded binary broadcast channel
\begin{equation}
\begin{split}
    Y_1 &= X \oplus S_1,\\
    Y_2 &= Y_1 \oplus S_2 = X\oplus S_1 \oplus S_2,
\end{split}
\end{equation}
where $(X, S_1, S_2, Y_1, Y_2)$ all take binary values $\{0, 1\}$, and $S_1 \sim \Bern(p_1)$, $S_2 \sim \Bern(p_2)$ are independent with $0<p_1 < p_2 < \frac{1}{2}$. Without the state estimation tasks, the capacity region of such a channel is 
\begin{equation}
\begin{split}
    R_2 &< 1 - H_2(\alpha * \tilde{p}_2),\\
    R_1 &< H_2(\alpha * p_1) - H_2(p_1),
\end{split}
\end{equation}
with $0\le\alpha\le \frac{1}{2}$, $\tilde{p}_2 = p_1 * p_2 = p_1(1-p_2) + (1-p_1)p_2$\cite{el2011network}. The achieving random variables are given by $U\sim \Bern(\frac{1}{2})$, $X = U\oplus V$ with $V\sim \Bern(\alpha)$.

Suppose Decoder 1 and 2 are interested in $S_1$, $S_2$, respectively, and the distortion function is the Hamming distance $d_H$. Clearly, Decoder 1 can achieve $D_1 = 0$. With the above choice of $(U,X)$, the estimator at Decoder 2 depends on $(U, Y_2)$, or equivalently, on $V \oplus S_1 \oplus S_2$. Note that $V\oplus S_1 \sim \Bern(\alpha * p_1)$ is independent of $S_2$, and thus the minimum achieved distortion at Decoder 2, in this case, is given by 
\begin{equation}
    D_{2,V_2=\varnothing} = \min \{\alpha * p_1, p_2 \}
\end{equation}
due to Lemma~\ref{lemma:bernoptest}. If we choose $V_1 = \varnothing$, $V_2 = V$ and keep $(U, X)$ as before, the resulting region from \eqref{eq:rdregion} becomes
\begin{equation}
\begin{split}
    R_2 & < I(U; Y_2) - I(V_2; X | U, Y_2)\\
    & =  1 - H_2(\alpha * \tilde{p}_2) - H(V | U, Y_2)\\
    & = 1 - H_2(\alpha * \tilde{p}_2) - H(V | V\oplus S_1 \oplus S_2)\\
    &= 1 - H_2(\alpha * \tilde{p}_2) - H(V) \\
    &\hspace{10mm}- H(V\oplus S_1 \oplus S_2 | V) + H(  V\oplus S_1 \oplus S_2)\\
    &= 1 - H_2(\alpha * \tilde{p}_2) - H_2(\alpha) - H_2(\tilde{p}_2) + H(\alpha * \tilde{p}_2)\\
    & = 1 - H_2(\alpha) - H_2(\tilde{p}_2),\\
    R_1 & < I(X; V_2, Y_1 |U)\\
    &= H(V,Y_1 | U) - H(V,Y_1|U,X)\\
    &= H(V) + H(Y_1 | U,V) - H(Y_1 | X)\\
    &= H(V) = H_2(\alpha).
\end{split}
\end{equation}
As Decoder 2 knows $V_2=V$ now, the estimator depends only on $S_1 \oplus S_2$ and the achieved distortion is
\begin{equation}
    D_{2,V_2 = V} = \min \{p_1, p_2\} = p_1.
\end{equation}
To achieve such distortion, one should set $\alpha = 0$ in the first coding scheme, i.e., no message is sent to Decoder 1, because $\alpha * p_1\ge p_1$ for any $0\le\alpha\le\frac{1}{2}$. In contrast, with proper selections of $\alpha$, one can achieve nonzero rate pairs $(R_1, R_2)$ in the second scheme. As an example, we set $p_1 = \frac{1}{20}$, $p_2 = \frac{1}{10}$. The obtained regions are illustrated in Fig.~\ref{fig:bcbin}.

\begin{figure}
    \centering
    \includegraphics[scale=0.8]{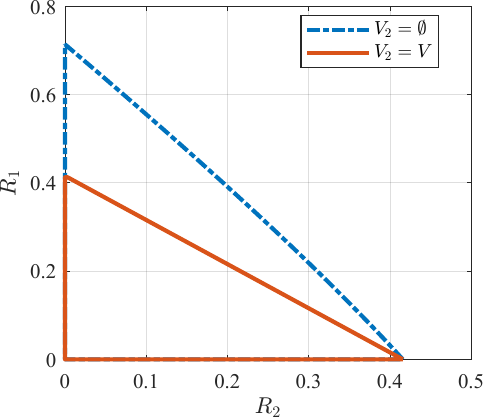}
    \caption{Achieved rate regions of the degraded binary broadcast channel with different choices of $V_2$.}
    \label{fig:bcbin}
\end{figure}

\section{Multiple Access Channel}
\subsection{Channel Model}

The studied \ac{sddmmac} for \ac{jsmc} is illustrated in Fig.~\ref{fig:macmodel}.
The individual private messages $M_1\in \calM_1 = [2^{nR_1}]$, $M_2 \in \calM_2 = [2^{nR_2}]$ are encoded into the transmit signals $X_1^n$ and $X_2^n$, respectively. The channel is characterized by the transition distribution $P_{Y|X_1X_2S}$. The feedbacks $Z_1$, $Z_2$ are assumed to be deterministic functions of $Y$, i.e., $Z_1 = \varphi_1(Y)$ and $Z_2 = \varphi_2(Y)$. Upon receiving $Y^n$, the decoder jointly recovers both messages and channel state.
The channel state $S_i$ at time step $i$ is generated \ac{iid} according to $P_S$. Both encoders additionally have access to their respective \ac{sit} $S_1^{i-1}$ and $S_2^{i-1}$ strictly causally, following $P_{S_1S_2|S}$ at each time step. The system consists of the following components:
\begin{enumerate}
    \item Encoding functions $q_{e,k}^n$ with $q_{e,k,i}: \calM_k \times \calS_k^{i-1} \times \calZ_k^{i-1} \to \calX_k$ at Encoder $k \in \{1,2\}, \forall i \in [n]$;
    \item Message decoding function $q_d: \calY^n \to \calM_1 \times \calM_2$;
    \item State estimation function $h^n$with $h_i: \calY^n \to \hat{\calS}, \forall i \in [n]$.
\end{enumerate}

\begin{figure}[h]
\centering
\includegraphics[scale=0.9]{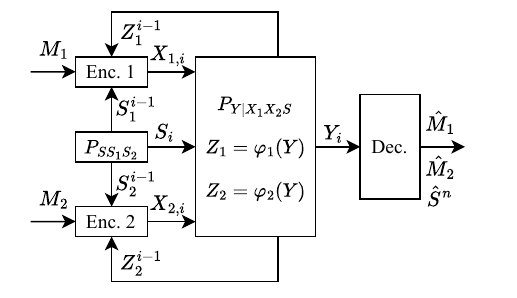}
\caption{The multiple access \ac{jsmc} model.}\label{fig:macmodel}
\end{figure}

Let $d$ be the distortion function as before.
One can define the message decoding error probability $P_e^{(n)}$, averaged distortion $D^{(n)}$, and the corresponding \ac{cd} region.
\begin{definition}
    A tuple of rates and distortion $(R_1, R_2, D)$ is said to be achievable if there exists a sequence of $(2^{nR_1}, 2^{nR_2}, n)$ codes such that
    \begin{align}
            \lim_{n\rightarrow \infty} P_e^{(n)} = 0,\ \limsup_{n\rightarrow \infty} D^{(n)} \le D.
    \end{align}
    The \ac{cd} region of the \ac{sddmmac} in Fig.~\ref{fig:macmodel}, denoted as $\calC^{\MAC}(D)$, is the closure of the set of $(R_1, R_2)$, such that $(R_1, R_2, D)$ is achievable.
\end{definition}

\subsection{Main Results}

We first present an achievable \ac{rd} region for the \ac{sddmmac} model. The coding scheme combines Willems' strategy\cite{willems1983achievable, ahmadipour2023information} to enable partial collaboration between both encoders and the proposed simultaneous decoding approach. 
% In particular, due to the presence of feedback, the two encoders are able to transmit part of their private messages to each other, presented in the codeword $(W_1, W_2)$, which in turn plays a role of shared information. In addition, the two encoders compress the \ac{sit}, feedback, and codewords from the previous block as the state descriptions $(V_1, V_2)$, and encode them into the transmit signals along with the fresh private messages, i.e., the part that is unknown to the other encoder. The decoder performs joint decoding of all the codewords and state descriptions in the backward direction, resulting in an improved region compared to our previous work\cite{Li2411:Achievable}, where the sequential decoding method is applied.

Let $\calP_D^{\MAC}$ be the set of all random variables $(U,W_1,W_2,X_1, X_2, V_1, V_2)$ and functions $h: \calU\times \calW_1 \allowbreak \times \calW_2\times \calX_1\times \calX_2\times \calV_1\times \calV_2\times \calY \to \hat{\calS}$ such that
\begin{equation}
\begin{split}
    &\hspace{-2mm}\calP_D^{\MAC} \triangleq \left\{ (U, W_1,W_2,X_1,X_2, V_1, V_2, h)\right|\\
    &\hspace{-2mm}P_{SS_1S_2UW_1W_2X_1X_2YZ_1Z_2V_1V_2}(s,s_1,s_2,u,w_1,w_2,\\
    &x_1,x_2,y,z_1,z_2,v_1,v_2) =P_{SS_1S_2}(s,s_1,s_2)P_{U}(u)\\
    &P_{W_1X_1|U}(w_1,x_1|u)P_{W_2X_2|U}(w_2,x_2|u)\\
    &P_{Y | X_1X_2 S} ( y | x_1,x_2 , s)\mathbbm{1}\{z_1=\varphi_1(y)\} \mathbbm{1}\{z_2=\varphi_2(y)\}\\
    &P_{V_1|UW_1W_2X_1S_1Z_1}(v_1|u,w_1,w_2,x_1,s_1,z_1)\\
    &P_{V_2|UW_1W_2X_2S_2Z_2}(v_2|u,w_1,w_2,x_2,s_2,z_2); \\
    &\left. \expcs{}{d(S, h(U,W_1,W_2, X_1,X_2, V_1, V_2, Y))} \le D \right \}.
\end{split}
\end{equation}
The derived achievable \ac{rd} region is given by $\calR^{\MAC}(D) = co(\calR_i^{\MAC}(D))$, i.e., the convex hull of $\calR_i^{\MAC}(D)$, which contains all nonnegative rate pairs $(R_1, R_2)$ such that
\begin{equation}\label{eq:macrdregion}
\begin{split}
    &\hspace{-1cm} R_1 < I(X_1, V_1; Y | U,W_1,W_2,X_2,V_2) \\
    &+ I(W_1; Z_2 | S_2,U,W_2,X_2) \\
    &-I(V_1; S_1, Z_1 | U, W_1, W_2, X_1)\\
    &\hspace{-1cm} R_2 < I(X_2, V_2; Y | U,W_1,W_2,X_1,V_1)  \\
    &+ I(W_2; Z_1 | S_1,U,W_1,X_1)\\
    &- I(V_2; S_2, Z_2 | U, W_1, W_2, X_2)\\
    &\hspace{-1cm} R_1 + R_2 <  I(X_1, X_2, V_1, V_2 ; Y | U,W_1,W_2) \\
    &+  I(W_1; Z_2 | S_2,U,W_2,X_2)\\
    &+ I(W_2; Z_1 |S_1, U,W_1,X_1)\\
    &-  I(V_1; S_1, Z_1 | U, W_1, W_2, X_1)\\
    &- I(V_2; S_2, Z_2 | U, W_1, W_2, X_2)\\
    &\hspace{-1cm}R_1 + R_2<  I(U,W_1,W_2,X_1,X_2,V_1,V_2; Y) \\
    &- I(V_1; S_1, Z_1 | U, W_1, W_2, X_1)\\
    &- I(V_2; S_2, Z_2 | U, W_1, W_2, X_2)
\end{split}
\end{equation}
% \begin{equation}\label{eq:macrdregion}
% \begin{split}
%     &\calR_i^{\MAC}(D) \triangleq \bigcup_{(U, W_1,W_2,X_1,X_2, V_1, V_2, h) \in \calP_D^{\MAC}}\\
%     &\hspace{1cm} \brcur{
%         \begin{array}{ll}
%                            & R_1\ge 0, R_2 \ge 0,\\
%                            & R_1 < I(X_1, V_1; Y | U,W_1,W_2,X_2,V_2) \\
%                            &\hspace{1cm}+ I(W_1; Z_2 | S_2,U,W_2,X_2) - I(V_1; S_1, Z_1 | U, W_1, W_2, X_1)\\
%                            & R_2 < I(X_2, V_2; Y | U,W_1,W_2,X_1,V_1)  \\
%                            &\hspace{1cm}+ I(W_2; Z_1 | S_1,U,W_1,X_1)- I(V_2; S_2, Z_2 | U, W_1, W_2, X_2)\\
%         (R_1, R_2):        & R_1 + R_2 <  I(X_1, X_2, V_1, V_2 ; Y | U,W_1,W_2) \\
%                            &\hspace{1cm}+  I(W_1; Z_2 | S_2,U,W_2,X_2)+ I(W_2; Z_1 |S_1, U,W_1,X_1)\\
%                            &\hspace{1cm}-  I(V_1; S_1, Z_1 | U, W_1, W_2, X_1)- I(V_2; S_2, Z_2 | U, W_1, W_2, X_2),\\
%                            & R_1 + R_2<  I(U,W_1,W_2,X_1,X_2,V_1,V_2; Y) \\
%                            &\hspace{1cm}- I(V_1; S_1, Z_1 | U, W_1, W_2, X_1)- I(V_2; S_2, Z_2 | U, W_1, W_2, X_2)\\
%         \end{array}
%         }.
% \end{split}
% \end{equation}
for $(U, W_1,W_2,X_1,X_2, V_1, V_2, h) \in \calP_D^{\MAC}$.

The basic idea of the proposed coding scheme is twofold. The first is to enable partial collaboration between encoders through the feedback links. We follow Willems' coding scheme\cite{willems1983achievable, ahmadipour2023information}, splitting the total data rate at Encoder $k\in \{1,2\}$ into two parts $R_k = R_k'+R_k''$, where $R_k'$ is the data rate that can be decoded by the other encoder, represented by the random variable $W_k$ and constrained by $I(W_{\bar{k}}; Z_k | U, W_k, X_k, S_k)$ for $\bar{k} \neq k$. This decoded shared information is then represented by the auxiliary random variable $U$ for collaboration. The second idea is to split $R_k''$ further into the data rate for the remaining private messages and that for the state descriptions. The compression of the \ac{sit} $S_k$ with the knowledge of $Z_k$ and all the other codewords from the previous block is exactly the same as before, represented by $V_k$. The decoder performs the backward simultaneous decoding of all the codewords and state descriptions, which is beneficial because the rate bounds on many error events are redundant. Such a simultaneous decoding approach is also applied in Willems' scheme, and thus, our \ac{rd} region $\calR^{\MAC}(D)$ reduces to  Willems' rate region when $S_1=S_2 =\varnothing$.

\begin{theorem}\label{thm:macinner}
    The \ac{cd} region for the \ac{sddmmac} in Fig.~\ref{fig:macmodel} satisfies
    \begin{equation}
        \calR^{\MAC}(D) \subseteq \calC^{\MAC}(D).
    \end{equation}
\end{theorem}
\begin{proof}
    See 
    \if \doublecol1
        \cite[Appendix~F]{li2025simultaneousdecodingapproachjoint}
    \else
        Appendix~\ref{app:mac-thm}
    \fi
    for detailed proof. A sketch of the coding scheme is provided here.

    In each block $b$, Encoder $k\in \{1,2\}$ splits the private messages $m_{k,b} = (m_{k,b}', m_{k,b}'')$ with $m_{k,b}'\in [2^{nR_k'}]$, $m_{k,b}''\in [2^{nR_k''}]$ and $R_k=R_k'+R_k''$. Both encoders generate $2^{n(R_1'+R_2')}$ sequences $u^n(b) = u^n(m_{1,b-1}', m_{2,b-1}')$ from $P^n_U$, For each $u^n(b)$, Encoder $k$ generates $2^{nR_k'}$ sequences of $w_k^n(b) = w_k^n(m_{k,b}'|m_{1,b-1}', m_{2,b-1}')$ from $P^n_{W_k|U}$, $2^{n(R_k''+R_{sk})}$ sequences of $x_k^n(b)=x_k^n(m_{k,b}'', l_{k,b-1}|m_{k,b}', \allowbreak m_{1,b-1}', m_{2,b-1}')$ given each $(u^n(b), w_k^n(b))$ from $P^n_{X_k|W_kU}$, and $v_k^n(b)=v_k^n(l_{k,b}|m_{k,b}'',l_{k,b-1},\allowbreak m_{1,b}', m_{2,b}',m_{1,b-1}', m_{2,b-1}')$ given each $(u^n(b), w_1^n(b), \allowbreak w_2^n(b), x_k^n(b))$ from $P^n_{V_k|UW_1W_2X_k}$, with $l_{k,b}\in [2^{nR_{sk}}]$.

    Encoder $k$ first decodes $m_{\bar{k},b}'$ for $\bar{k}\neq k$ through the feedback link by looking for a unique jointly typical sequence from $\typset{}{P_{UW_1W_2X_kS_kZ_k}}$, requiring
    \begin{equation}\label{eq:macrk'}
        R_{\bar{k}}' < I(W_{\bar{k}}; W_k, X_k, S_k, Z_k |U) = I(W_{\bar{k}}; Z_k | U, W_k, X_k, S_k).
    \end{equation}
    Then it encodes the state description $l_{k,b}$ by finding a jointly typical sequence from $\typset{}{P_{V_kUW_1W_2X_kS_kZ_k}}$, guaranteed by
    \begin{equation}
        R_{sk} > I(V_k; S_k , Z_k | U, W_1, W_2, X_k).
    \end{equation}
    It sends $x^n_k(b)$ in each block, and in the last blocks, it doesn't encode any message and only sends the state description. The other encoder applies the same steps.

    The decoder starts from the last blocks by decoding the state descriptions. The correct decoding in the last blocks is proved in the appendix. Afterward, it decodes all the codewords simultaneously in the backward direction, requiring 
    \begin{equation*}
    \begin{split}
        R_k'' + R_{sk} &< I(X_k,V_k; Y | U,W_1,W_2,X_{\bar{k}},V_{\bar{k}})\\
        R_1'' + R_2'' + R_{s1} + R_{s2} & < I(X_1,X_2,V_1,V_2; Y | U,W_1,W_2)\\
        R_1+R_2+R_{s1} + R_{s2} &< I(U,W_1,W_2,X_1,X_2,V_1,V_2; Y).
    \end{split}
    \end{equation*}
    The channel state is estimated as $h(u_i(b), w_{1,i}(b), w_{2,i}(b),  \allowbreak x_{1,i}(b), x_{2,i}(b), v_{1,i}(b), v_{2,i}(b), y_i(b))$ for $i\in [n]$, $b\in [B]$.

    Combining the requirements above, we obtain the \ac{rd} region $\calR_i^{\MAC}(D)$. By applying the time-sharing strategy, the final region should be a convex hull of $\calR_i^{\MAC}(D)$.
\end{proof}

\begin{prop}\label{prop:macrd}
    $\calR^{\MAC}(D)$ has the following properties:
    \begin{enumerate}
        \item If $D_1 \ge D_2$, then $\calR^{\MAC}(D_1) \supseteq \calR^{\MAC}(D_2)$;
        \item $\calR^{\MAC}(D)$ is a convex set;
        \item To exhaust $\calR^{\MAC}(D)$, the estimation function $h$ is given by the optimal estimator
        \begin{equation*}
        \begin{split}
            &h^*(u, w_1,w_2,x_1,x_2,v_1, v_2, y) = \\
            &\argmin_{\hat{s}\in \hat{\calS}} \sum_{s\in\calS} P_{S|UW_1W_2X_1X_2V_1V_2Y}(s|u, w_1,w_2,x_1,x_2,\\
            &\hspace{5.5cm}v_1, v_2, y) d(s,\hat{s}).
        \end{split}
        \end{equation*}
    \end{enumerate}
\end{prop}
\begin{proof}
    The proof is conducted similarly to before and thus is omitted here.
\end{proof}

% \begin{remark}
    It is worth mentioning that in the absence of feedback, $\calR^{\MAC}(D)$ coincides with the rate region in\cite{li2012multiple} for $D=\infty$ by setting $U= W_1 = W_2 = \varnothing$, resulting in
    \begin{equation}\label{eq:macrd-nofb}
    \begin{split}
        R_1&< I(X_1, V_1; Y | X_2, V_2) - I(V_1; S_1 | X_1),\\
        R_2&< I(X_2, V_2; Y | X_1, V_1) - I(V_2; S_2 | X_2),\\
        R_1 + R_2&< I(X_1, X_2, V_1, V_2; Y) \\
        &\hspace{1cm}- I(V_1; S_1 | X_1)- I(V_2; S_2 | X_2).
    \end{split}
    \end{equation}
    Unlike the point-to-point and \ac{sddmbc} models, where state descriptions can be set to $\varnothing$ for $D=\infty$, $(V_1, V_2)$ in $\calP_D^{\MAC}$ can still be used to convey the state information to assist message decoding, as mentioned in \cite{li2012multiple, lapidoth2012multiple}, and is shown to enlarge the achievable rate region. 
    % Therefore, we observe the double usages of $(V_1, V_2)$, not only enhancing the communication performance but also improving the state estimation distortion.

% \begin{remark}
    In\cite{Li2411:Achievable}, we present an achievable region combining the sequential decoding\cite{choudhuri2013causal} and the strategy in\cite{lapidoth2012multiple}. Two-stage state descriptions are designed there, where the first stage description at each encoder compresses only the \ac{sit} and feedback from the previous block and is decoded before messages to enhance the communication performance, following the same idea in\cite{lapidoth2012multiple}. The second stage description, designed in the same way as\cite{choudhuri2013causal} using random binning at each encoder, can only be decoded after all the other codewords and thus used purely for the state estimation task. Providing the new achievable region $\calR^{\MAC}(D)$, which can be shown larger than the previous one in \cite{Li2411:Achievable}, the proposed coding scheme unifies the usage of state description for both communication performance enhancement and state estimation due to the simultaneous decoding process.
% \end{remark}

A simple outer bound can be obtained by enabling the full collaboration between encoders. In this case, the channel is equivalent to the point-to-point case, and the sum rate is constrained similarly to \eqref{eq:p2prdfunc}. Let the corresponding region be
\begin{equation}\label{eq:macouter}
\begin{split}
    &\calO^{\MAC}(D) \triangleq \{(R_1, R_2): R_1\ge 0, R_2\ge 0,\\
    &R_1 + R_2 < I(X_1, X_2 ; Y) - I(V_1, V_2; S_1, S_2 | X_1, X_2, Y)\}
\end{split}
\end{equation}
for all possible distributions $P_{X_1X_2}$, $P_{V_1V_2 | S_1S_2X_1X_2Z_1Z_2}$, and estimators $h$ such that
\begin{equation}
    \expcs{}{d(S, h( X_1,X_2, V_1, V_2, Y))} \le D.
\end{equation}
\begin{theorem}\label{thm:macsumouter}
    The \ac{cd} region for the \ac{sddmmac} in Fig.~\ref{fig:macmodel} satisfies
    \begin{equation}
        \calC^{\MAC}(D) \subseteq \calO^{\MAC}(D).
    \end{equation}
\end{theorem}
\begin{proof}
    The proof of converse is conducted in the same way as the point-to-point case in\cite{bross2017rate} and is omitted here.
\end{proof}

\subsection{Examples}

\subsubsection{Monostatic-uplink ISAC}

Another important scenario of \ac{isac} is the monostatic-uplink system\cite{Li2411:Achievable, liu2022survey}, where the \ac{bs} sends signals purely for monostatic sensing and simultaneously decodes the uplink signal from a \ac{ue}. This system can be modeled as a \ac{sddmmac} with two encoders as the \ac{bs} transmitter and the \ac{ue}, the decoder as the \ac{bs} receiver. By recognizing the sensing signal as $X_2$, uplink signal as $X_1$, $Y=(Y',X_2)$ with the normal channel output $Y'$, $S_1 = S_2 = Z_1 = \varnothing$, and $Z_2=Y'$, it shows from Theorem~\ref{thm:macinner} that the relationship between uplink data rate and sensing distortion $C_{\mathrm{MU}}(D)$ is 
\begin{equation*}
\begin{split}
    C_{\mathrm{MU}}(D) &\triangleq \max_{\substack{P_U,P_{X_1W_1|U}, P_{X_2|U}:\\ \expcs{}{d(S, h^*(U, W_1, X_1, X_2, Y'))} \le D}}I(X_1, W_1  ; Y' |U, X_2)\\
    &= \max_{\substack{P_U,P_{X_1|U}, P_{X_2|U}:\\ \expcs{}{d(S, h^*(U, X_1, X_2, Y'))} \le D}}I(X_1 ; Y' |U, X_2)
\end{split}
\end{equation*}
by setting $W_1 = X_1$. In fact, $C_{\mathrm{MU}}(D)$ is the optimal \ac{cd} function in this case, and the proof of converse is as follows.

\begin{proof}
    First, we show that $C_{\mathrm{MU}}(D)$ is a non-decreasing and concave function in $D$. The non-decreasing property is obvious from Proposition~\ref{prop:macrd}. Let $(U_1, X_{1,1}, X_{2,1})$ and $(U_2, X_{1,2}, X_{2,2})$ be two sets of random variables achieving $C_{\mathrm{MU}}(D_1)$ and $C_{\mathrm{MU}}(D_2)$ respectively. Introducing an independent time-sharing random variable $Q\in \{1,2\}\sim P_Q$ with $P_Q(1) = \lambda$ and defining $\tilde{U} = (Q, U_Q)$, the achieved distortion $D$ by $(\tilde{U}, X_{1,Q}, X_{2,Q})$ is $D = \lambda D_1 + (1-\lambda)D_2$ due to the linearity of the distortion constraint, and $X_{1,Q} - \tilde{U} - X_{2,Q}$ forms a Markov chain. Then, we have
    \begin{align*}
        C_{\mathrm{MU}}(D) &= \max_{\substack{P_U,P_{X_1|U}, P_{X_2|U}:\\ \expcs{}{d(S, h^*(U, X_1, X_2, Y'))} \le D}}I(X_1 ; Y' |U, X_2)\\
        &\ge I(X_{1,Q} ; Y' |\tilde{U}, X_{2,Q})\\
        &= I(X_{1,Q} ; Y' | U_Q, X_{2,Q}, Q)\\
        &= \lambda I(X_{1,1} ; Y' | U_1, X_{2,1})\\
        &\hspace{1cm}+ (1-\lambda) I(X_{1,2} ; Y' | U_2, X_{2,2})\\
        &= \lambda C_{\mathrm{MU}}(D_1) + (1-\lambda) C_{\mathrm{MU}}(D_2). 
    \end{align*}
    The converse is proved by applying the Fano's inequality
    \begin{subequations}
    \begin{align}
        R_1 - \epsilon_n &\le \frac{1}{n}I(M_1; Y^n)\\
        &=\frac{1}{n} \sum_{i=1}^n I(M_1; Y_i | Y^{i-1})\\
        &= \frac{1}{n}\sum_{i=1}^n I(M_1; Y_i | Y'^{i-1}, X_2^{i-1}, X_{2,i}) \label{subeq:macexamp1}\\
        &=\frac{1}{n} \sum_{i=1}^n I(M_1, X_{1,i}; Y'_i | Y'^{i-1}, X_2^{i-1}, X_{2,i})\\
        &\le\frac{1}{n} \sum_{i=1}^n I(X_{1,i}; Y'_i | U_i, X_{2,i})\label{subeq:macexamp2}\\
        &\le \frac{1}{n} \sum_{i=1}^n C_{\mathrm{MU}}\br{\expcs{}{d(S_i, h^*(U_i, X_{1,i}, X_{2,i}, Y'_i)}}\label{subeq:macexamp3}\\
        &\le  C_{\mathrm{MU}}\br{\frac{1}{n} \sum_{i=1}^n\expcs{}{d(S_i, h^*(U_i, X_{1,i}, X_{2,i}, Y'_i)}}\label{subeq:macexamp4}\\
        &\le C_{\mathrm{MU}}(D).
    \end{align}
    \end{subequations}
    where we define $U_i = Y'^{i-1}$ such that $X_{1,i} - U_i - X_{2,i}$ forms a Markov chain, \eqref{subeq:macexamp1} follows since $X_{2,i}$ is a function of $Y'^{i-1}$, \eqref{subeq:macexamp2} holds due to the Markov chain $(M_1, X_2^{i-1}, Y'^{i-1}) - (X_{1,i}, X_{2,i}) - Y_i'$, \eqref{subeq:macexamp3} follows the definition of $C_{\mathrm{MU}}(D)$, \eqref{subeq:macexamp4} is due to the concavity of $C_{\mathrm{MU}}(D)$. The last inequality holds by the non-decreasing property of $C_{\mathrm{MU}}(D)$, and by noting the Markov chain $S_i - (U_i, X_{1,i}, X_{2,i}, Y'_i) - (X_2^{n\backslash i}, Y'^n_{i+1})$ and applying Lemma~\ref{lemma:markovest} such that
    \begin{equation}
    \begin{split}
        D\ge& \frac{1}{n}\sum_{i=1}^n\expcs{}{d(S_i, h_i(X_2^n, Y'^n))}\\
        \ge& \frac{1}{n} \sum_{i=1}^n\expcs{}{d(S_i, h^*(U_i, X_{1,i}, X_{2,i}, Y'_i))}.
    \end{split}
    \end{equation}
    This completes the proof.
\end{proof}

\subsubsection{Quadratic-Gaussian Case}

Let the \ac{sddmmac} be characterized by,
\begin{equation}
    Y=X_1+X_2+S+W,
\end{equation}
and $S_1 = S_2 = S_T = S + W_T$ where $S\sim \calN(0,Q)$, $W\sim \calN(0, N)$, $W_T\sim \calN(0,N_T)$. The feedback is assumed to be perfectly available at both encoders, i.e., $Z_1=Z_2=Y$. The distortion function is again the squared error, and the transmit signal power is constrained by $\expcs{}{|X_1|^2}\le P_1$ and $\expcs{}{|X_2|^2}\le P_2$, respectively. We shall choose random variables as follows
\begin{equation}\label{eq:mac_rv}
\begin{split}
    &W_1=X_1,\ W_2=X_2,\  U\sim \calN(0, 1),\\
    &X_1 = \sqrt{\alpha_1P_1} U + U_1, U_1\sim \calN(0, (1-\alpha_1) P_1),  0\le \alpha_1 \le 1,\\
    &X_2 = \sqrt{\alpha_2P_2} U + U_2, U_2\sim \calN(0, (1-\alpha_2) P_2),  0\le \alpha_2 \le 1,\\
    &V_1 = S_T + E_1,\ E_1 \sim \calN(0, d_1^2),\\
    &V_2 = S_T + E_2,\ E_2 \sim \calN(0, d_2^2),
\end{split}
\end{equation}
and $U, U_1, U_2, E_1, E_2$ are all independent. The desired distortion is evaluated as
\begin{equation}
\begin{split}
    D &= \var{}{S|U,X_1,X_2,V_1,V_2,Y}\\
    &= \var{}{S|V_1,V_2,S+W},
\end{split}
\end{equation}
which only depends on $d_1^2$ and $d_2^2$. As a comparison, the time-sharing strategy is also considered, where at each time instance, only one encoder is active, communicating at a rate given by \eqref{eq:cqg}, and the other one sends deterministic signals. Moreover, an outer bound is also calculated from \eqref{eq:macouter} by enabling full cooperation between both encoders. Note that the corresponding power constraint changes to
\begin{equation}
\begin{split}
    \expcs{}{|X_1 + X_2|^2} &= \expcs{}{|X_1|^2} + \expcs{}{|X_2|^2} + 2\expcs{}{X_1X_2} \\
    & \le P_1 + P_2 + 2\sqrt{P_1P_2}.
\end{split}
\end{equation}
The computed regions are plotted in Fig.~\ref{fig:macgaussian} by setting $P_1=P_2=5$, $Q=N=1$, $N_T=0.3$ and $d_1^2=d_2^2=1.7333$ such that $D=0.35$. It shows that the region obtained by $\calR^{\MAC}(D)$ based on the chosen random variables \eqref{eq:mac_rv}, denoted as `Proposed', is strictly better than the region with time-sharing strategy, denoted as `TS', even at the corner point where $R_1=0$ or $R_2=0$, highlighting the benefits of encoder cooperation. On the other hand, the outer bound of the sum rate achieved by full cooperation, labeled as `OUTER', improves the region further.

\begin{figure}
    \centering
    \includegraphics[scale=0.8]{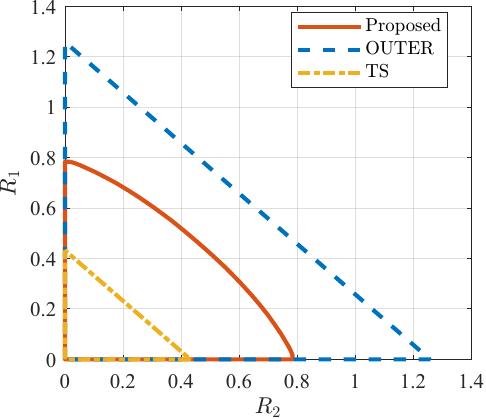}
    \caption{Comparison of different regions for the quadratic-Gaussian MAC.}
    \label{fig:macgaussian}
\end{figure}

\subsubsection{Double Usages of Single State Description}

We revisit Example 3 in \cite{lapidoth2012multiple} and show that the single state description can improve not only the communication data rate but also the state estimation distortion. The channel is characterized by
\begin{equation}
\begin{split}
    &Y = (Y_1, Y_2),\\
    &S = (T_1, T_2), \ S_1 = \varnothing, \ S_2 = S,\\
    &Y_1 = X_1 \oplus T_{X_2},\\
    &Y_2 = X_2,
\end{split}
\end{equation}
and no feedback is present for either encoder. All the random variables are binary, taking value from $\{0,1\}$, and $H(T_0) = H(T_1) = \frac{1}{2}$. It shows that by setting $V_1 = \varnothing$, $V_2 = T_{X_2}$ and $X_1, X_2$ independent and uniformly distributed in \eqref{eq:macrd-nofb}, one can achieve $(R_1, R_2) = (1, \frac{1}{2})$. If the \ac{sit} is ignored and $V_1=V_2=\varnothing$, we have
\begin{align*}
    R_1 &< I(X_1; Y | X_2) =H(Y_1|X_2) - H(Y_1|X_1,X_2)\\
    &=(1-p_2) H(X_1 \oplus T_0) + p_2H(X_1\oplus T_1)\\
    &\hspace{1cm}- (1-p_2)H(T_0) - p_2H(T_1)\\
    &\le 1 - \frac{1}{2} = \frac{1}{2},
\end{align*}
where $X_2 \sim \mathrm{Bern}(p_2)$, and the last inequality holds because the maximum entropy of a binary random variable is $1$. Hence, $(R_1, R_2) = (1, \frac{1}{2})$ is not achievable in this case, showing that compressing the state can enhance communication performance even if no state estimation task is required, as already studied in\cite{lapidoth2012multiple, li2012multiple}.

Next, suppose that the decoder wishes to estimate the state $S$ as $\hS = (\hat{T}_0, \hat{T}_1)$, and the distortion function is given by
\begin{equation}
    d(S, \hS) = d_H(T_0, \hat{T}_0) + d_H(T_1, \hat{T}_1),
\end{equation}
with $d_H$ the Hamming distance. Note that the estimator from \eqref{eq:macrdregion} now depends on $(X_1, X_2, V_1,V_2,Y)$. It's easy to show that with the two schemes above, the decoder can only recover $T_{X_2}$ at each time while the other state is unknown, whose optimal estimator is thus random guessing, which yields a strictly positive distortion $D>0$. If further $V_2 = S = (T_0, T_1)$, we have from Theorem~\ref{thm:macinner}
\begin{equation*}
\begin{split}
    R_1 & < I(X_1; Y | X_2, T_0, T_1)\\
    & = H(Y | X_2, T_0, T_1)\\
    & = H(Y_1 | X_2, T_0, T_1)\\
    & = H(X_1),\\
    R_2 &< I(X_2, T_0, T_1; Y | X_1) - I(T_0,T_1; T_0,T_1|X_2)\\
    &= H(Y_1,Y_2|X_1) - H(T_0,T_1)\\
    &= H(X_2) + H(Y_1|X_1,X_2) - 1,\\
    R_1 + R_2 &< I(X_1,X_2,T_0,T_1; Y) - I(T_0,T_1; T_0,T_1|X_2)\\
    &= H(Y) - 1\\
    &= H(X_2) + H(Y_1 | X_2) - 1.
\end{split}
\end{equation*}
Let $X_1, X_2$ both follow $\mathrm{Bern}(\frac{1}{2})$ and be independent, one can achieve the region $\{R_1 < 1, R_2 < \frac{1}{2}, R_1+R_2 < 1\}$ since $H(Y_1 | X_2) = \frac{1}{2}(H(X_1\oplus T_0) + H(X_1\oplus T_1)) = 1$. Some rate pairs in this region are still not achievable, e.g., $(R_1, R_2) = (\frac{2}{3}, \frac{1}{3})$, by the case of ignoring \ac{sit}, and the distortion now becomes $D=0$. This demonstrates the two-fold benefits brought by the single-state description, providing a degree of freedom to reach different goals.

\section{Conclusion}\label{sec:conclusion}
This work proposes a coding framework for the \ac{jsmc} problem based on the backward simultaneous decoding approach, achieving optimality for point-to-point channels and improving existing regions for multi-user channels. The advantages brought by the proposed scheme are validated through both theoretical analysis and numerical examples.
% The backward simultaneous decoding approach without random binning is proposed for \ac{jsmc} problem over single- and multi-user channels in this manuscript. 
% The coding scheme is shown to be optimal for point-to-point channels and especially beneficial for \ac{sddmbc} and \ac{sddmmac} compared to the sequential decoding approach. Numerical and analytical examples are provided to study the behavior of the derived achieved functions/regions with different settings. Additionally, with a focus on \ac{isac} systems, the fundamental limits of interesting scenarios are also investigated.
However, several limitations warrant further investigation. First, the current analysis assumes memoryless channels. Extending to channels with memory is crucial for more practical scenarios. Second, the optimality of successive refinement coding in \ac{sddmbc} remains unproven, and the necessity of leaking private messages to the weak user to improve state estimation performance should be discussed. Third, although even the capacity regions are unknown for \ac{sddmmac} with feedback or \ac{sit}, special cases of quadratic-Gaussian channels have been well studied\cite{ozarow1984capacity,li2012multiple}. Future works may investigate certain scenarios under which the optimal \ac{cd} region can be derived.
% despite deriving achievable regions under feedback, exact capacity bounds are elusive—a gap that may be addressed by leveraging quadratic-Gaussian models to derive closed-form C-D functions, mirroring classical capacity results./
% These directions not only address theoretical gaps but also align with emerging applications in ISAC and IoT, where unified communication-sensing systems demand rigorous trade-off analyses between data rate efficiency and state reconstruction accuracy.

% or


\appendices

\if \doublecol1

\else
    \newpage
    \section{}\label{app:p2p-thm}
    Fixing the set $\calP_D$ achieving $C(D)$:
\begin{equation}
\begin{split}
    \calP_D = & \{(X,V,h) | P_{SS_TXYZV}(s, s_T, x, y,z, v) = P_{S_TS}(s_T,s)P_{X}(x) \\
    &\quad \cdot  P_{Y|XS}(y|x,s)\mathbbm{1}\{z=\phi(y)\}P_{V|XS_TZ}(v|x,s_T,z); \expcs{}{d(S, h(X,V,Y))}\le D\}.
\end{split}
\end{equation}
It is assumed $I(X;Y)> \mu >0$ with a positive number $\mu$. Otherwise, the encoder can transmit nothing. In each block, the \ac{sit}, feedback, and the codeword from the previous block are compressed, described by the variable $V$, which is then encoded into the codeword at the current block along with the fresh message. In the last block, the encoder only sends the state description without any fresh message. The decoder performs the backward simultaneous message and state description decoding. Starting from the last block, it first decodes the compressed state information of the second last block and then decodes the message and state description of the previous block jointly until the first block. The state description, together with the channel input and output, is used to estimate the channel state at each block. 
% This coding scheme differs from the traditional ones in \cite{bross2017rate, choudhuri2013causal}, where the message is decoded before the state description decoding. 
The detailed proof is elaborated as follows. We assume the transmission happens in $B + 1$ blocks.

\subsection{Codebook generation}
For block $b\in [B]$, randomly and independently generate $2^{n(R+R_s)}$ sequences $x^n(m_b, l_{b-1})$ according to $P_X^n$, each of which is assigned by the indices $m_b\in [2^{nR}]$ and $l_{b-1}\in [2^{nR_s}]$. For each $x^n(m_b, l_{b-1})$, randomly generate $2^{nR_s}$ sequences $v^n(l_b | m_b, l_{b-1})$ \ac{iid} according to $P_{V|X}^n$, each of which is assigned by the indices $l_b \in [2^{nR_s}]$.

For the last block $b=B+1$, we only generate $2^{nR_s}$ length-$n'$ sequences $x^{n'}(1, l_{B})$ \ac{iid} according to $P_X^n$ with $n' = \frac{nR_s}{\mu}$.
    
\subsection{Encoding}
Let $l_0=1$. In each block $b\in[B]$, having the message $m_{b}$ to be encoded, $x_i(m_b, l_{b-1})$ for $i\in[n]$ are transmitted.
At the end of block $b$, upon knowing $(s_T^n(b), z^n(b))$, find an index $\hat{l}_b$ such that
\begin{equation}\label{appeq:compression}
    \br{s_T^n(b), z^n(b), x^n(m_b, l_{b-1}), v^n(\hat{l}_b | m_b, l_{b-1})}\in \typset{}{P_{S_TZXV}}.
\end{equation}
If there are multiple such indices, select one of them randomly. If there is no such index, set $\hat{l}_b=1$.
In the last block $B+1$, the encoder transmits $x_i(1, l_{B})$ for $i\in [n']$.

\subsection{Decoding}
In the last block $B+1$, upon receiving $y^{n'}(B+1)$, the decoder finds the unique index $\hat{l}_B$ such that
\begin{equation}
    \br{x^{n'}(1, \hat{l}_B), y^{n'}(B+1)}\in \calT^{(n')}_{\epsilon}(P_{XY}).
\end{equation}
If there are multiple such indices, select one of them randomly. If there is no such index, set $\hat{l}_B=1$.
For the subsequent blocks from $b=B$ down to $b=1$, by assuming $l_b$ is decoded correctly in advance, the decoder finds the unique indices $(\hat{m}_b, \hat{l}_{b-1})$ jointly such that
\begin{equation}
    \br{x^n(\hat{m}_b, \hat{l}_{b-1}), v^n(l_b | \hat{m}_b, \hat{l}_{b-1}),  y^n(b)}\in \typset{}{P_{XVY}}.
\end{equation}
If there are multiple such index tuples, select randomly from them. If there is no such tuple, set $(\hat{m}_b, \hat{l}_{b-1}) = (1,1)$.
The state is estimated as
\begin{equation}
\hat{s}_i(b)=h(x_i(\hat{m}_b, \hat{l}_{b-1}), v_i(l_b | \hat{m}_b, \hat{l}_{b-1}), y_i(b))
\end{equation}
for $i\in[n]$.

\subsection{Analysis}

For the compression step \eqref{appeq:compression} at the encoder, we have with probability close to $1$,
\begin{equation}
    \br{s_T^n(b), z^n(b), x^n(m_b, l_{b-1})}\in \typset{}{P_{S_TZX}}.
\end{equation}
By covering lemma \cite{el2011network}, when
\begin{equation}
    R_s > I(V; S_T, Z | X)
\end{equation}
the probability of the error event that no index $l_b$ satisfies \eqref{appeq:compression} can be bounded to $0$, i.e.,
\begin{equation}
   \lim_{n\to \infty} \pr{\forall \hat{l}_b \in[2^{nR_s}]: \br{s_T^n(b), z^n(b), x^n(m_b, l_{b-1}), v^n(\hat{l}_b | m_b, l_{b-1})}\notin \typset{}{P_{S_TZXV}}} = 0.
\end{equation}

At the decoder, the definition of $n'$ guarantees the correct decoding of $\hat{l}_B$, i.e., the rate for transmitting $l_B$ satisfies
\begin{equation}
    \frac{nR_s}{n'} = \mu < I(X;Y).
\end{equation}
The joint decoding of $(\hat{m}_b, \hat{l}_{b-1})$ requires 
\begin{equation}
    R + R_s < I(X, V; Y),
\end{equation}
such that
\begin{equation}
    \lim_{n\to \infty}\pr{\exists (\hat{m}_b, \hat{l}_{b-1}) \neq (m_b, l_{b-1}): \br{x^n(\hat{m}_b, \hat{l}_{b-1}), v^n(l_b | \hat{m}_b, \hat{l}_{b-1}),  y^n(b)}\in \typset{}{P_{XVY}}} = 0.
\end{equation}

Consequently, the probability of encoding and decoding error tends to zero as $n \to \infty$ if
\begin{equation}
\begin{split}
    R &< I(X, V; Y) - I(V; S_T, Z | X)\\
   &= I(X; Y) + I(V; Y |X) -  I(V; S_T, Z | X) \\
   &= I(X; Y) - I(V; S_T | X,Y),
\end{split}
\end{equation}
where the last step follows the Markov chain $V-(S_T, X, Z) - Y$ and $Z$ is a deterministic function of $Y$. The achievable data rate becomes 
\begin{equation}
    \frac{nBR}{nB + n'} = \frac{R}{ 1 + \frac{n'}{nB}} \overset{B\to \infty}{=}  R,
\end{equation}
and the state estimation distortion is also bounded by $D$ as $n\to \infty$ and $B \to \infty$.

    \section{}\label{app:p2p-prop}
    The non-decreasing property is obvious because for $D_1<D_2$, we have $\calP_{D_1}\subseteq \calP_{D_2}$. The concavity can be verified by the time-sharing strategy. Suppose $(X_1, V_1, h_1)$ and $(X_2,V_2,h_2)$ achieve $C(D_1)$ and $C(D_2)$ respectively. Let $Q\in\{1,2\}\sim P_Q$ be an independent time-sharing random variable with $P_Q(1)=\lambda$, and define $\tilde{V} = (V_Q, Q)$.
% The joint distribution involving $Q$ is given by
% \begin{equation*}
% \begin{split}
%     &P_Q(q)P_{S_TS}(s_T,s)P_{X|Q}(x|q)P_{Y|XS}(y|x,s)\\
%     &\qquad\cdot\mathbbm{1}\{z=\phi(y)\}P_{V_Q|XS_TZQ}(v|x,s_T,z,q), \quad \forall q\in \{1,2\}.
% \end{split}
% \end{equation*}
The distortion achieved by $(X_Q, \tilde{V}, h_Q)$ is then $D = \lambda D_1 + (1-\lambda) D_2$ due to the linearity of the distortion constraint and thus $(X_Q, \tilde{V}, h_Q)\in\calP_{D}$.
% because the distortion constraint is linear in the joint distribution and the estimator, the distortion attained by time-sharing is $D = \lambda D_1 + (1-\lambda) D_2$. 
We then have
\begin{equation}
\begin{split}
    C(D) &= \max_{(X,V,h)\in\calP_{D}} I(X;Y) - I(V;S_T|X,Y)\\
    &\ge I(X_Q; Y) - I(\tilde{V};S_T|X_Q,Y)\\
    &\ge I(X_Q; Y |Q) - I(V_Q;S_T|X_Q,Y,Q)\\
    &= \lambda(I(X_1; Y) - I(V_1;S_T|X_1,Y)) + (1-\lambda)(I(X_2; Y) - I(V_2;S_T|X_2,Y))\\
    &=\lambda C(D_1) + (1-\lambda)C(D_2),
\end{split}
\end{equation}
where we adopt the Markov chain $Q-X-Y$ for the third line. Hence, it is obtained that $C(D)\ge \lambda C(D_1) + (1-\lambda)C(D_2)$ for $D = \lambda D_1 + (1-\lambda) D_2$. From Lemma~\ref{lemma:optest} and the non-decreasing property of $C(D)$, the expressions for the optimal estimator and the minimum achievable distortion are simply inferred. If $S_T$ is statistically independent of $S$, we have the Markov chain $V- (X, Y) - S$ such that removing the dependency on $V$ in the optimal estimator doesn't change the achieved distortion due to Lemma~\ref{lemma:markovest}. It thus suffices to set $V=\varnothing$ as $I(V; S_T | X, Y) \ge 0$.
% The cardinality of $V$ can be determined according to Lemma~\ref{lemma:support}. The expression $I(V; S_T | X, Y) = H(S_T|X,Y) - H(S_T| X, Y, V)$ and the distortion constraint can be viewed as $|\calS_T| + 1$ functions of $P_{S_T | XYV}$ for each fixed $X=x$ and $Y=y$: $P_{S_T|XY}(\cdot | x,y)$ for all but one $s_T\in \calS_T$, $H(S_T | X=x, Y=y, V)$, and $\expcs{}{d(S, h(V,X,Y)) | X=x, Y=y}$. This concludes $|\calV|\le |\calS_T|+1$. 

    \section{}\label{app:deg-bc-scc}
    The following lemma states the condition of correct simultaneous decoding, which will be used later for the analysis.
\begin{lemma}\label{applemma:bc}
    Let $(U,X,V_1,V_2, Y_1) \sim P_{UXV_1V_2Y_1}$.
    Suppose $(u^n, v_2^n, y_1^n) \in \typset{}{P_{UV_2Y_1}}$, $X^n$ are generated \ac{iid} from $P_{X|U}(\cdot | u_i)$ and $V_1^n$ are generated \ac{iid} from $P_{V_1|UXV_2}(\cdot | u_i, x_i, v_{2,i})$. There exists $\delta(\epsilon)$ that tends to zero as $\epsilon \to 0$ such that 
    \begin{equation}
        \Pr\brcur{(X^n, V_1^n, u^n, v_2^n, y_1^n) \in \typset{}{P_{XV_1UV_2Y_1}}} \le 2^{-n[I(X;V_2, Y_1 |U) + I(V_1; Y_1 | U,X,V_2)-3\delta(\epsilon)]}.
    \end{equation}
\end{lemma}
\begin{proof}
    \begin{equation}
    \begin{split}
        \Pr&\brcur{(X^n, V_1^n, u^n, v_2^n, y_1^n) \in \typset{}{P_{XV_1UV_2Y_1}}}\\
        &=\sum_{(x^n, v_1^n)\in \typset{}{P_{XV_1UV_2Y_1}|u^n,v_2^n,y_1^n}} P^n_{XV_1 | UV_2Y_1}(x^n, v_1^n | u^n, v_2^n, y_1^n)\\
        &=\sum_{(x^n, v_1^n)\in \typset{}{P_{XV_1UV_2Y_1}|u^n,v_2^n,y_1^n}} P^n_{X|U}(x^n | u^n) P^n_{V_1 | UXV_2} (v_1^n | u^n, x^n, v_2^n)\\
        &\le 2^{n(H(X,V_1 | U,V_2,Y_1) + \delta(\epsilon))} 2^{-n(H(X|U)-\delta(\epsilon))} 2^{-n(H(V_1|UXV_2)-\delta(\epsilon))}\\
        &= 2^{-n[I(X; V_2,Y_1 | U) + I(V_1; Y_1 | U,X,V_2) - 3\delta(\epsilon)]}.
    \end{split}
    \end{equation}
\end{proof}

We fix the distributions of $(U, X, V_1, V_2, h_1, h_2)\in \pdd^{\BC}$ achieving $\calR^{\BC}(D_1, D_2)$:
\begin{equation}
\begin{split}
    \pdd^{\BC} &= \left\{ (U, X, V_1, V_2, h_1, h_2)\right| P_{SS_TUXY_1Y_2ZV_1V_2}(s,s_T,u,x,y_1,y_2,z,v_1,v_2)=P_{SS_T}(s,s_T)\\
    & \cdot P_{UX}(u,x)P_{Y_1 Y_2 | X S} (y_1, y_2 | x , s)\mathbbm{1}\{z=\psi(y_1, y_2))\} P_{V_1V_2|UXS_TZ}(v_1,v_2|u,x,s_T,z); \\
    &\left. \expcs{}{d_1(S, h_1(U, X, V_1, V_2, Y_1))} \le D_1, \expcs{}{d_2(S, h_2(U, V_2, Y_2))} \le D_2 \right \},
\end{split}
\end{equation}
Let $I(X;Y_1 | U) > \mu_1>0$ and $I(U; Y_2) > \mu_2 >0$. Note that if either $I(X;Y_1 | U)=0$ or $I(U; Y_2)=0$, the channel reduces to a single user point-to-point channel and the remaining proof follows Appendix~\ref{app:p2p-thm}. The transmission is assumed to happen in $B+1$ blocks.

\subsection{Codebook Generation}
    For block $b\in [B]$, with $m_{1,b}\in [2^{nR_1}]$, $l_{1,b}\in [2^{nR_{s1}}]$, $m_{2,b}\in [2^{nR_2}]$, $l_{2,b}\in [2^{nR_{s2}}]$, generate
    \begin{itemize}
        \item $2^{n(R_2+R_{s2})}$ sequences of $u^n(m_{2,b}, l_{2,b-1})$ \ac{iid} according to $P^n_U$;
        \item $2^{n(R_1+R_{s1})}$ sequences of $x^n(m_{1,b}, l_{1,b-1} | m_{2,b}, l_{2,b-1})$ \ac{iid} according to $P^n_{X|U}$ given each $u^n(m_{2,b}, l_{2,b-1})$;
        \item $2^{nR_{s2}}$ sequences of $v_2^n(l_{2,b}|m_{2,b}, l_{2,b-1})$ \ac{iid} according to $P^n_{V_2|U}$ given each $u^n(m_{2,b}, l_{2,b-1})$;
        \item $2^{nR_{s1}}$ sequences of $v_1^n(l_{1,b}| m_{1,b}, l_{1,b-1},m_{2,b}, l_{2,b-1}, l_{2,b})$ \ac{iid} according to $P^n_{V_1|UXV_2}$ given each $u^n(m_{2,b}, l_{2,b-1})$, $x^n(m_{1,b},\allowbreak l_{1,b-1} | m_{2,b}, l_{2,b-1})$ and $v_2^n(l_{2,b}|m_{2,b}, l_{2,b-1})$.
    \end{itemize}
    
    % generate $2^{n(R_2+R_{s2})}$ sequences of $u^n(m_{2,b}, l_{2,b-1})$ \ac{iid} according to $P^n_U$, $2^{n(R_1+R_{s1})}$ sequences of $x^n(m_{1,b},\allowbreak l_{1,b-1} | m_{2,b}, l_{2,b-1})$ \ac{iid} according to $P^n_{X|U}$ given each $u^n(m_{2,b}, l_{2,b-1})$, $2^{nR_{s2}}$ sequences of $v_2^n(l_{2,b}|m_{2,b}, l_{2,b-1})$ \ac{iid} according to $P^n_{V_2|U}$ given each $u^n(m_{2,b}, l_{2,b-1})$ and $2^{nR_{s1}}$ sequences of $v_1^n(l_{1,b}| m_{1,b}, l_{1,b-1},m_{2,b}, l_{2,b-1}, l_{2,b})$ \ac{iid} according to $P^n_{V_1|UXV_2}$ given each $u^n(m_{2,b}, l_{2,b-1})$, $x^n(m_{1,b}, l_{1,b-1} | m_{2,b}, l_{2,b-1})$ and $v_2^n(l_{2,b}|m_{2,b}, l_{2,b-1})$ with $m_{1,b}\in [2^{nR_1}]$, $l_{1,b}\in [2^{nR_{s1}}]$, $m_{2,b}\in [2^{nR_2}]$, $l_{2,b}\in [2^{nR_{s2}}]$. 

    For the last block $B+1$, let $n' = \max (\frac{nR_{s1}}{\mu_1}, \frac{nR_{s2}}{\mu_2})$, generate 
    \begin{itemize}
        \item $2^{nR_{s2}}$ length-$n'$ sequences $u^{n'}(1,l_{2,B})$ \ac{iid} according to $P^{n'}_U$;
        \item $2^{nR_{s1}}$ length-$n'$ sequences $x^{n'}(1,l_{1,B} | 1,l_{2,B})$ \ac{iid} according to $P^{n'}_{X|U}$ given each $u^{n'}(1,l_{2,B})$.
    \end{itemize}
    % $2^{nR_{s2}}$ length-$n'$ sequences $u^{n'}(1,l_{2,B})$ \ac{iid} according to $P^{n'}_U$ and $2^{nR_{s1}}$ length-$n'$ sequences $x^{n'}(1,l_{1,B} | 1,l_{2,B})$ \ac{iid} according to $P^{n'}_{X|U}$ given each $u^{n'}(1,l_{2,B})$ with $n' = \max (\frac{nR_{s1}}{\mu_1}, \frac{nR_{s2}}{\mu_2})$.
    
\subsection{Encoding}
    Let $l_{1,0} = l_{2,0} =1$. In each block $b\in[B]$, having the messages $m_{1,b}, m_{2,b}$ to be encoded, the encoder sends $x_i(m_{1,b}, l_{1,b-1}  |\allowbreak m_{2,b}, l_{2,b-1})$ for $i\in [n]$.
    At the end of block $b$, upon knowing $(s^n_T(b), z^n(b))$, find an index tuple $(\hat{l}_{1,b}, \hat{l}_{2,b})$ such that
    \begin{equation}
    \begin{split}
    &\left(s^n_T(b), z^n(b), x^n(m_{1,b}, l_{1,b-1} | m_{2,b}, l_{2,b-1}), v_2^n(\hat{l}_{2,b}|  m_{2,b}, l_{2,b-1}), u^n(m_{2,b}, l_{2,b-1}) \right) \in \typset{}{P_{S_TZXV_2U}},\\
    &\left(s^n_T(b), z^n(b), x^n(m_{1,b}, l_{1,b-1} | m_{2,b}, l_{2,b-1}), v_2^n(l_{2,b}|  m_{2,b}, l_{2,b-1}), \right.\\
    &\hspace{3.1cm} \left.v_1^n(\hat{l}_{1,b} | m_{1,b}, l_{1,b-1},m_{2,b}, l_{2,b-1}, l_{2,b}),u^n(m_{2,b}, l_{2,b-1}) \right) \in \typset{}{P_{S_TZXV_2V_1U}}.
    \end{split}
    \end{equation}
    If there are multiple such index tuples, select one randomly from them. If there is no such index tuple, set $(\hat{l}_{1,b}, \hat{l}_{2,b})=(1,1)$.
    In the last block $B+1$, the encoder transmits $x_i(1, l_{1,B} | 1, l_{2,B})$ for $i\in [n']$.

\subsection{Decoding at Decoder 2}
    The decoding steps at the weaker user are the same as the point-to-point case in Appendix~\ref{app:p2p-thm}. In the beginning, due to the choice of $n'$, the state description $l_{2,B}$ can be decoded reliably because
    \begin{equation}\label{appeq:R2cond1}
        \frac{nR_{s2}}{n'} \le \mu_2 < I(U;Y_2).
    \end{equation}
    % The subsequent decoding error can be bounded to 0 as $n\to \infty$ if
    % \begin{equation}\label{appeq:R2cond2}
    %     R_2 + R_{s2} < I(U, V_2; Y_2).
    % \end{equation}

\subsection{Decoding at Decoder 1}
    The stronger user first looks for $(\hat{l}_{1,B}, \hat{l}_{2,B})$ in the last block such that
    \begin{equation}
        \br{u^{n'}(1, \hat{l}_{2,B}), x^{n'}(1, \hat{l}_{1,B} | 1, \hat{l}_{2,B}), y_1^{n'}(B+1)} \in \calT^{(n')}_{\epsilon}(P_{UXY_1}).
    \end{equation}
    If there are multiple such index tuples, select one from them randomly. If there is no such index tuple, set $(\hat{l}_{1,B}, \hat{l}_{2,B}) = (1,1)$. This is guaranteed to be decoded correctly. To see this, first note that the probability of the event that $\hat{l}_{2,B}$ is wrong tends to $0$ if 
    \begin{equation}
        \frac{nR_{s2}}{n'} \le I(U; Y_1)
    \end{equation}
    which is automatically satisfied by \eqref{appeq:R2cond1} because the channel is degraded. Furthermore, given the correct $\hat{l}_{2,B}$, the correct decoding of $\hat{l}_{1,B}$ can also be guaranteed because
    \begin{equation}
        \frac{nR_{s1}}{n'} \le \mu_1 < I(X; Y_1 |U).
    \end{equation}
    
    Suppose $(l_{1,b}, l_{2,b})$ are correctly decoded, the stronger user then finds the unique index set $(\hat{m}_{1,b}, \hat{m}_{2,b}, \hat{l}_{1,b-1}, \hat{l}_{2,b-1})$ from block $b=B$ down to $b=1$ such that
    \begin{equation}
    \begin{split}
        % &\br{u^n(m_{2,b}, l_{2,b-1}), v^n_2(l_{2,b}, \hat{k}_{2,b} |m_{2,b}, l_{2,b-1}), x^n(m_{1,b}, l_{1,b-1} | m_{2,b}, l_{2,b-1}), y_1^n(b)}\in \typset{2}{P_{UV_2XY_1}}\\
        &\left(u^n(\hat{m}_{2,b}, \hat{l}_{2,b-1}),v^n_2(l_{2,b} |\hat{m}_{2,b}, \hat{l}_{2,b-1}), x^n(\hat{m}_{1,b}, \hat{l}_{1,b-1} | \hat{m}_{2,b}, \hat{l}_{2,b-1}), \right. \left. v_1^n(l_{1,b} | \hat{m}_{1,b}, \hat{l}_{1,b-1},\hat{m}_{2,b}, \hat{l}_{2,b-1}, l_{2,b}), y_1^n(b) \right )\\
        &\hspace{13cm}\in \typset{}{P_{UV_2XV_1Y_1}}.
    \end{split}
    \end{equation}
    If there are multiple such index tuples, select one from them randomly. If there is no such index tuple, set $(\hat{m}_{1,b}, \hat{m}_{2,b},\allowbreak \hat{l}_{1,b-1}, \hat{l}_{2,b-1})=(1,1,1,1)$. The state is then estimated by applying $h_1$ to all decoded codewords for all blocks.
    % The error probability of decoding $(\hat{m}_{1,b}, \hat{l}_{1,b-1})$ given $(\hat{m}_{2,b}, \hat{l}_{2,b-1})$ being correct can be bounded to $0$ as $n\to \infty$ when
    
    % according to Lemma~\ref{applemma:bc}. It can be shown that the probabilities of the other error events all tend to zero as $n\to \infty$ if \eqref{appeq:R2cond2} and \eqref{appeq:R1cond1} satisfy.

\subsection{Analysis}
    We have that with probability close to $1$
    \begin{equation}
        (s^n_T(b), z^n(b), x^n(m_{1,b}, l_{1,b-1} | m_{2,b}, l_{2,b-1}), u^n(m_{2,b}, l_{2,b-1})) \in \typset{}{P_{S_TZXU}}.
    \end{equation}
    With the covering lemma\cite{el2011network}, if 
    \begin{equation}
    \begin{split}
        &R_{s2} > I(V_2; S_T, Z, X | U),\\
        &R_{s1} > I(V_1; S_T, Z | U, X, V_2).
    \end{split}
    \end{equation}
    we have
    \begin{equation}
    \begin{split}
    % &\lim_{n\to \infty}\pr{\forall \hat{l}_{2,b}\in [2^{R_{s2}}]: \left(s^n_T(b), z^n(b), x^n(m_{1,b}, l_{1,b-1} | m_{2,b}, l_{2,b-1}), v_2^n(\hat{l}_{2,b}|  m_{2,b}, l_{2,b-1}), u^n(m_{2,b}, l_{2,b-1}) \right) \notin \typset{}{P_{S_TZXV_2U}}} = 0,\\
    &\lim_{n \to \infty}\mathrm{Pr}\left\{\forall (\hat{l}_{1,b}, \hat{l}_{2,b})\in [2^{R_{s1}}]\times [2^{R_{s2}}]: \left(s^n_T(b), z^n(b), x^n(m_{1,b}, l_{1,b-1} | m_{2,b}, l_{2,b-1}), v_2^n(\hat{l}_{2,b}|  m_{2,b}, l_{2,b-1}), \right. \right.\\
    &\hspace{3.1cm}\left. \left.v_1^n(\hat{l}_{1,b} | m_{1,b}, l_{1,b-1},m_{2,b}, l_{2,b-1}, \hat{l}_{2,b}),u^n(m_{2,b}, l_{2,b-1}) \right) \notin \typset{}{P_{S_TZXV_2V_1U}}\right\} = 0.
    \end{split}
    \end{equation}
    The decoding at Decoder 2 requires
    \begin{equation}\label{appeq:bc-rate2}
        R_2 + R_{s2} < I(U, V_2; Y_2).
    \end{equation}
    as the point-to-point case. For Decoder 1, the error event
    \begin{equation}
    \begin{split}
        &\exists (\hat{m}_{1,b}, \hat{m}_{2,b}, \hat{l}_{1,b-1}, \hat{l}_{2,b-1}) \neq (m_{1,b}, m_{2,b}, l_{1,b-1}, l_{2,b-1}):\\
        &\left(u^n(\hat{m}_{2,b}, \hat{l}_{2,b-1}),v^n_2(l_{2,b} |\hat{m}_{2,b}, \hat{l}_{2,b-1}), x^n(\hat{m}_{1,b}, \hat{l}_{1,b-1} | \hat{m}_{2,b}, \hat{l}_{2,b-1}), v_1^n(l_{1,b} | \hat{m}_{1,b}, \hat{l}_{1,b-1},\hat{m}_{2,b}, \hat{l}_{2,b-1}, l_{2,b}), y_1^n(b) \right )\\
        &\hspace{10cm}\in \typset{}{P_{UV_2XV_1Y_1}}
    \end{split}
    \end{equation}
    can be split into multiple sub-events. It shows that all the other sub-events are redundant if
    \begin{equation}\label{appeq:error-event1}
    \begin{split}
        &\mathrm{Pr}\left\{\exists (\hat{m}_{1,b}, \hat{l}_{1,b-1}) \neq (m_{1,b}, l_{1,b-1}), (\hat{m}_{2,b}, \hat{l}_{2,b-1}) = (m_{2,b}, l_{2,b-1}):\right.\\
        &\left(u^n(\hat{m}_{2,b}, \hat{l}_{2,b-1}),v^n_2(l_{2,b} |\hat{m}_{2,b}, \hat{l}_{2,b-1}), x^n(\hat{m}_{1,b}, \hat{l}_{1,b-1} | \hat{m}_{2,b}, \hat{l}_{2,b-1}), v_1^n(l_{1,b} | \hat{m}_{1,b}, \hat{l}_{1,b-1},\hat{m}_{2,b}, \hat{l}_{2,b-1}, l_{2,b}), y_1^n(b) \right )\\
        &\left.\hspace{10cm}\in \typset{}{P_{UV_2XV_1Y_1}}\right \} = 0
    \end{split}
    \end{equation}
    and \eqref{appeq:bc-rate2} satisfies. \eqref{appeq:error-event1} holds true if
    \begin{equation}\label{appeq:R1cond1}
    R_1 + R_{s1} < I(X; V_2, Y_1 | U) + I(V_1; Y_1 | U,X,V_2)
    \end{equation}
    according to Lemma~\ref{applemma:bc}.
    Combining the above requirements on $(R_1, R_{s1}, R_2, R_{s2})$, we have
    \begin{equation}
    \begin{split}
        R_1 &< I(X; V_2, Y_1|U) + I(V_1; Y_1 | U,X,V_2) - R_{s1}\\
            &< I(X; V_2, Y_1|U) + I(V_1; Y_1 | U,X,V_2) - I(V_1; S_T, Z | U, X, V_2)\\
            &= I(X; V_2, Y_1|U) - I(V_1; S_T, Z| U, X, V_2, Y_1)
    \end{split}
    \end{equation}
    and
    \begin{equation}
    \begin{split}
        R_2 &< I(U, V_2; Y_2) - R_{s2}\\
            &< I(U, V_2; Y_2) - I(V_2; S_T, Z, X | U)\\
            &= I(U; Y_2) - I(V_2; S_T, Z, X| U, Y_2)
    \end{split}
    \end{equation}
    due to the Markov chain $(V_1, V_2) - (U, X, S_T, Z) - (Y_1, Y_2)$.
    The encoding and decoding error probability can then be bounded to zero as $n\rightarrow \infty$, such that the distortion $(D_1, D_2)$ can be achieved, which concludes the proof. 
    
    \section{}\label{app:bc-prop}
    The first property follows due to $\pdd^{\BC} \subseteq \pddf{D_1'}{D_2'}^{\BC}$. Denoting the random variables achieving $\calR^{\BC}(D_1', D_2') $ and $\calR^{\BC}(D_1'', D_2'') $ as $(U_1, X_1, V_{1,1}, V_{2,1}, h_{1,1}, h_{2,1})$ and $(U_2, X_2, V_{1,2}, V_{2,2}, h_{1,2}, h_{2,2})$, respectively, similar to Proposition~\ref{prop:cdcausalconcave}, we introduce a time-sharing random variable $Q\in \{1,2\}$ with $P_Q(1)=\lambda$ and define $\tilde{U} = (U_Q,Q)$. The achieved distortions are $D_1=\lambda D_1' + (1-\lambda) D_1''$,  $D_2=\lambda D_2' + (1-\lambda) D_2''$, so $(\tilde{U}, X_Q, V_{1,Q}, V_{2,Q}, h_{1,Q}, h_{2,Q}) \in \pdd^{\BC}$. The achieved rates are
\begin{equation}
\begin{split}
    R_2 &= \lambda \br{I(U_1; Y_2) - I(V_{2,1}; S_T, Z, X_1| U_1, Y_2)} +(1-\lambda) \br{I(U_2; Y_2) - I(V_{2,2}; S_T, Z, X_2| U_2, Y_2)}\\
    &= I(U_Q;Y_2 | Q) - I(V_{2,Q}; S_T, Z,X_Q| U_Q,Q,Y_2)\\
    &\le I(U_Q,Q;Y_2)  - I(V_{2,Q}; S_T, Z,X_Q| U_Q,Q,Y_2)\\
    &= I(\tilde{U}; Y_2) - I(V_{2,Q}; S_T, Z,X_Q| U_Q,Q,Y_2),
\end{split}
\end{equation}
and 
\begin{equation}
    R_1 \le I(X_Q; V_{2,Q}, Y_1| \tilde{U}) - I(V_{1,Q}; S_T, Z| X_Q, \tilde{U}, V_{2,Q}, Y_1),
\end{equation}
meaning $(R_1,R_2)\in \calR^{\BC}(D_1, D_2)$, from which we can infer property 3) by letting $D_1' = D_1''$ and $D_2' = D_2''$. Because of property 1) and Lemma~\ref{lemma:optest}, we obtain the expressions for both optimal estimators.

    \section{}\label{app:deg-bc-outer}
    Define the following auxiliary random variables:
\begin{equation}\label{eq:deg-bc-auxrv}
\begin{split}
    U_i &= (M_2, Y_1^{i-1}, Y_2^{i-1})\\
    % U_{2,i} &= (U_i, S_T^{i-1}) = (M_2, Y_1^{i-1}, Y_2^{i-1}, S_T^{i-1})\\
    % U_{1,i} &= (M_1, U_{2,i}) = (M_1, M_2, Y_1^{i-1}, Y_2^{i-1}, S_T^{i-1})\\
    V_{2,i} &= (Y_1^{i-1}, Y_2^{n\backslash i})\\
    V_{1,i} &= (S_T^{i-1}, Y_{1,i+1}^n, V_{2,i}) = (S_T^{i-1}, Y_1^{n\backslash i}, Y_2^{n\backslash i}).
\end{split}
\end{equation}
It can be shown that the Markov chains $U_i - (X_i, S_{T,i}) - Y_{1,i} - Y_{2,i}$ and $(V_{1,i}, V_{2,i}) - (U_i, X_i, S_{T,i}, Z_i) - (Y_{1,i}, Y_{2,i})$ hold. Moreover, the estimators $h_k(Y_k^n)$ for $k=1,2$ can be expressed as
\begin{align}
    h_{2,i}(Y_2^n) &= h_2(i, U_i, V_{2,i}, Y_{2,i})\\
    h_{1,i}(Y_1^n) &= h_1(i, U_i, X_i, V_{1,i}, V_{2,i}, Y_{1,i}),
\end{align}
where we add $U_i$ and $(U_i, X_i, V_{2,i})$ into the dependency of $h_2$ and $h_1$ respectively without increasing the expected distortion. Introducing a time-sharing random variable $Q$ that is uniformly distributed over $[n]$ and independent of all the other random variables, and redefining $U=(U_Q, Q)$, $V_1 = V_{1,Q}$, $V_2=V_{2,Q}$, $S_T = S_{T,Q}$, $X=X_Q$, $Y_1 = Y_{1,Q}$, $Y_2=Y_{2,Q}$, $Z=Z_Q$, it is easy to show that the choice of $(U, X, V_1, V_2, h_1, h_2)$ is a member of $\pdd^{\BC}$. For Decoder 2, we have
\begin{subequations}
\begin{align}
    nR_2 - n\epsilon_n &\le I(M_2; Y_2^n)\\
    &= I(M_2, Y_1^n; Y_2^n) - I(Y_1^n; Y_2^n | M_2)\\
    &= \sum_{i=1}^n I(M_2, Y_1^n; Y_{2,i}|Y_2^{i-1}) - I(Y_{1,i}; Y_2^n | M_2, Y_1^{i-1})\\
    &\le \sum_{i=1}^n H(Y_{2,i}) - H(Y_{2,i} | M_2, Y_1^n, Y_2^{i-1}) - I(Y_{1,i}; Y_2^n | M_2, Y_1^{i-1})\\
    &= \sum_{i=1}^n H(Y_{2,i}) - H(Y_{2,i} | M_2, Y_1^i, Y_2^{i-1}) - I(Y_{1,i}; Y_2^n | M_2, Y_1^{i-1})\label{subeq:deg-bc-converse-outer-1}\\
    &= \sum_{i=1}^n I(M_2, Y_1^i, Y_2^{i-1}; Y_{2,i}) - I(Y_{1,i}; Y_2^n | M_2,  Y_1^{i-1})\\
    &= \sum_{i=1}^n I(M_2, Y_1^{i-1}, Y_2^{i-1}; Y_{2,i}) + I(Y_{1,i}; Y_{2,i} | M_2, Y_1^{i-1}, Y_2^{i-1})- I(Y_{1,i}; Y_2^n | M_2, Y_1^{i-1})\\
    &= \sum_{i=1}^n I(U_i; Y_{2,i}) + H(Y_{1,i} | M_2, Y_1^{i-1}, Y_2^{i-1}) - H( Y_{1,i} | M_2, Y_1^{i-1}, Y_2^{i-1}, Y_{2,i})\\
    &\hspace{1cm} - H(Y_{1,i} | M_2, Y_1^{i-1}) + H(Y_{1,i} | M_2, Y_1^{i-1}, Y_2^n)\notag\\
    &= \sum_{i=1}^n I(U_i; Y_{2,i}) - H(Y_{1,i} | M_2, Y_1^{i-1}, Y_2^{i-1}, Y_{2,i}) + H(Y_{1,i} | M_2,Y_1^{i-1}, Y_2^n)\\
    &= \sum_{i=1}^n I(U_i; Y_{2,i}) - I(Y_{1,i}; Y_{2,i+1}^n | M_2, Y_1^{i-1}, Y_2^{i-1}, Y_{2,i})\\
    &= \sum_{i=1}^n I(U_i; Y_{2,i}) - I(Y_{1,i}; V_{2,i} | U_i, Y_{2,i}),\label{subeq:deg-bc-converse-r2-end}
\end{align}
\end{subequations}
where \eqref{subeq:deg-bc-converse-outer-1} holds due to the Markov chain $Y_{2,i} - (M_2,Y_1^i,Y_2^{i-1}) - Y_{1,i+1}^n$. For Decoder 1, we have
\begin{subequations}
\begin{align}
    nR_1 - n\epsilon_n &\le I(M_1; Y_1^n|M_2)\\
    &= I(M_1, S_T^n; Y_1^n | M_2) - I(S_T^n; Y_1^n | M_1, M_2)\\
    &= \sum_{i=1}^n I(M_1, S_T^n; Y_{1,i} | M_2, Y_1^{i-1}) - I(S_{T,i}; Y_1^n | M_1, M_2, S_T^{i-1})\\
    &= \sum_{i=1}^n H(Y_{1,i} | M_2, Y_1^{i-1}, Y_2^{i-1}) - H(Y_{1,i} | M_1, M_2, Y_1^{i-1}, Y_2^{i-1}, S_T^n)- I(S_{T,i}; Y_1^n | M_1, M_2, S_T^{i-1})\\
    &= \sum_{i=1}^n I(M_1, S_T^i; Y_{1,i} | M_2, Y_1^{i-1}, Y_2^{i-1})- I(S_{T,i}; Y_1^n | M_1, M_2, S_T^{i-1})\\
    &= \sum_{i=1}^n I(M_1, S_T^{i-1}; Y_{1,i} | M_2, Y_1^{i-1}, Y_2^{i-1}) + I(S_{T,i}; Y_{1,i} | M_1, M_2, Y_1^{i-1}, Y_2^{i-1}, S_T^{i-1})- I(S_{T,i}; Y_1^n | M_1, M_2, S_T^{i-1})\\
    &= \sum_{i=1}^n I(X_i; Y_{1,i} | U_i) + I(S_{T,i}; Y_{1,i} | X_i, U_i)- I(S_{T,i}; Y_1^n |  M_1, M_2, S_T^{i-1})\\
    &= \sum_{i=1}^n I(X_i; Y_{1,i} | U_i) + H(S_{T,i}| X_i, U_i)\notag\\
    &\hspace{1cm} -  H(S_{T,i}|  X_i, U_i,  Y_{1,i}) - H(S_{T,i} | M_1, M_2, S_T^{i-1}) + H(S_{T,i} | M_1, M_2, S_T^{i-1},  Y_1^n)\\
    &\le \sum_{i=1}^n I(X_i; Y_{1,i} | U_i) - H(S_{T,i}|  X_i, U_i,  Y_{1,i}, Y_2^{n \backslash i}) + H(S_{T,i} | M_1, M_2, S_T^{i-1},  Y_1^n, Y_2^{n\backslash i})\label{subeq:deg-bc-converse-outer-2}\\
    &\le \sum_{i=1}^n I(X_i; Y_{1,i} | U_i) - H(S_{T,i}|  X_i, U_i,  Y_{1,i}, Y_2^{n \backslash i}) + H(S_{T,i} | X_i, U_i, Y_1^n, Y_2^{n\backslash i})\\
    &= \sum_{i=1}^n I(X_i; Y_{1,i} | U_i) - I(S_{T,i}; Y_1^{n\backslash i} |X_i, U_i, Y_{1,i}, Y_2^{n\backslash i})\\
    &= \sum_{i=1}^n I(X_i; Y_{1,i} | U_i) - I(S_{T,i}; V_{1,i} | X_i,  U_i, V_{2,i}, Y_{1,i})\label{subeq:deg-bc-converse-r1-end}
\end{align}
\end{subequations}
where \eqref{subeq:deg-bc-converse-outer-2} follows because $S_{T,i}$ is independent of $(U_i, X_i, M_1, M_2, S_T^{i-1})$, conditions don't increase entropy, and the channel is degraded. The sum rate is constrained by
\begin{subequations}
\begin{align}
    n(R_1 + R_2 - \epsilon_n) &\le I(M_1, M_2; Y_1^n)\\
    & = I(M_1, M_2, S_T^n ; Y_1^n) - I(S_T^n; Y_1^n | M_1, M_2)\\
    &=\sum_{i=1}^n I(M_1, M_2, S_T^n ; Y_{1,i} | Y_1^{i-1}) - I(S_{T,i}; Y_1^n | M_1, M_2, S_T^{i-1})\\
    &\le \sum_{i=1}^n I(M_1, M_2, S_T^i, Y_1^{i-1} ; Y_{1,i}) - I(S_{T,i}; Y_1^n | M_1, M_2, S_T^{i-1})\\
    &= \sum_{i=1}^n I(X_i, U_i; Y_{1,i}) + I(S_{T,i}; Y_{1,i} | U_i, X_i) - I(S_{T,i}; Y_1^n | M_1, M_2, S_T^{i-1})\\
    &\le \sum_{i=1}^n I(X_i; Y_{1,i}) - H(S_{T,i} | X_i, U_i, Y_{1,i}, Y_{2,i}) + H(S_{T,i} | U_i, X_i, S_T^{i-1}, Y_1^n, Y_2^n)\\
    &= \sum_{i=1}^n I(X_i; Y_{1,i}) - I(S_{T,i} ; S_T^{i-1}, Y_{1}^{n\backslash i}, Y_{2}^{n\backslash i} | X_i, U_i, Y_{1,i}, Y_{2,i})\\
    &= \sum_{i=1}^n I(X_i; Y_{1,i}) - I(S_{T,i}; V_{1,i} | X_i, U_i, Y_{1,i}, Y_{2,i}).
\end{align}
\end{subequations}
Leveraging the time-sharing random variable $Q$ and noting that it is independent of the other random variables, it is obtained
\begin{subequations}
\begin{align}
    R_2 - \epsilon_n &\le I(U_Q; Y_{2,Q} | Q) - I(Y_{1,Q}; V_{2,Q} | Q, U_Q, Y_{2,Q})\\
    &\le I(U; Y_2) - I(Y_1; V_2 | U, Y_2),\\
    R_1 - \epsilon_n &\le I(X_Q; Y_{1,Q} |Q, U_Q) - I(S_{T,Q}; V_{1,Q} | Q, U_Q, X_Q, V_{2,Q}, Y_{1,Q})\\
    &= I(X; Y_1 | U) - I(S_T; V_1 | U, X, V_2, Y_1),\\
    R_1 + R_2 - \epsilon_n &\le  I(X_Q, Y_{1,Q} | Q) - I(S_{T,Q}; V_{1,Q}, V_{2,Q} | Q, U_Q, X_Q, Y_{1,Q}, Y_{2,Q})\\
    &\le I(X, Y_1) - I(S_T; V_1, V_2 | U, X, Y_1, Y_2),
\end{align}
\end{subequations}
which concludes the proof.
    
    \section{}\label{app:mac-thm}
    The following two lemmas ensure the proper termination of the block Markov coding method. For simplicity, we write $\calP_D^{\MAC}$ here again:
\begin{equation}
\begin{split}
    \calP_D^{\MAC} &= \left\{ (U, W_1,W_2,X_1,X_2, V_1, V_2, h)\right|P_{SS_1S_2UW_1W_2X_1X_2YZ_1Z_2V_1V_2}(s,s_1,s_2,u,w_1,w_2,x_1,x_2,y,z_1,z_2,v_1,v_2)\\
    & =P_{SS_1S_2}(s,s_1,s_2)P_{U}(u)P_{W_1X_1|U}(w_1,x_1|u)P_{W_2X_2|U}(w_2,x_2|u)P_{Y | X_1X_2 S} ( y | x_1,x_2 , s)\mathbbm{1}\{z_1=\varphi_1(y))\} \\
    & \cdot \mathbbm{1}\{z_2=\varphi_2(y))\}P_{V_1|UW_1W_2X_1S_1Z_1}(v_1|u,w_1,w_2,x_1,s_1,z_1)P_{V_2|UW_1W_2X_2S_2Z_2}(v_2|u,w_1,w_2,x_2,s_2,z_2); \\
    &\left. \expcs{}{d(S, h(U,W_1,W_2, X_1,X_2, V_1, V_2, Y))} \le D \right \}.
\end{split}
\end{equation}
\begin{lemma}
    Let $(U,W_1,W_2,X_1,X_2,h) \in \calP_D^{\MAC}$ and $(R_1, R_2)\in\calR^{\MAC}(D)$, then $R_1 + R_2 < I(X_1,X_2;Y)$. Therefore, if $I(X_1; Y|X_2) = I(X_2; Y |X_1) = 0$, we have $R_1=R_2=0$.
\end{lemma}
\begin{proof}
\begin{subequations}
\begin{align}
    &R_1 + R_2\\
    &<  I(U,W_1,W_2,X_1,X_2,V_1,V_2; Y) - I(V_1; S_1, Z_1 | U, W_1, W_2, X_1)- I(V_2; S_2, Z_2 | U, W_1, W_2, X_2)\\
    &= I(X_1,X_2;Y) + I(U,W_1,W_2,V_1,V_2;Y | X_1,X_2) - I(V_1; S_1, Z_1 | U, W_1, W_2, X_1)- I(V_2; S_2, Z_2 | U, W_1, W_2, X_2)\\
    &= I(X_1,X_2;Y) + H(U,W_1,W_2,V_1, V_2| X_1,X_2) - H(U,W_1,W_2,V_1, V_2| X_1,X_2,Y) - H(V_1|U,W_1,W_2,X_1) \notag\\
    &\hspace{3cm} + H(V_1|U,W_1,W_2,X_1,S_1,Z_1)- H(V_2|U,W_1,W_2,X_2) + H(V_2|U,W_1,W_2,X_2,S_2,Z_2) \\
    &\le I(X_1,X_2;Y) + H(V_1, V_2|U,W_1,W_2, X_1,X_2) + H(U,W_1,W_2|X_1,X_2) - H(U,W_1,W_2|X_1,X_2,Y)\notag\\
    &\hspace{1cm} - H(V_1,V_2|U,W_1,W_2,X_1,X_2,Y) - H(V_1, V_2|U,W_1,W_2, X_1,X_2, Y) - H(V_1|U,W_1,W_2,X_1,X_2,V_2) \notag\\
    &\hspace{1cm} + H(V_1|U,W_1,W_2,X_1,S_1,Z_1)- H(V_2|U,W_1,W_2,X_1,X_2) + H(V_2|U,W_1,W_2,X_2,S_2,Z_2)\label{appsubeq:mac-lemma1-1}\\
    &= I(X_1,X_2;Y) - H(V_1, V_2|U,W_1,W_2, X_1,X_2, Y)+ H(V_1|U,W_1,W_2,X_1,S_1,Z_1) + H(V_2|U,W_1,W_2,X_2,S_2,Z_2)\\
    &= I(X_1,X_2;Y) - H(V_1, V_2|U,W_1,W_2, X_1,X_2, Y)+ H(V_1,V_2|U,W_1,W_2,X_1,X_2,S_1,S_2,Z_1,Z_2,Y)\label{appsubeq:mac-lemma1-2}\\
    &\le  I(X_1,X_2;Y),
\end{align}
\end{subequations}
where \eqref{appsubeq:mac-lemma1-1} follows because $(U,W_1,W_2)- (X_1,X_2) - Y$ forms a Markov chain and conditions do not increase entropy, and \eqref{appsubeq:mac-lemma1-2} holds due to the Markov chain $V_k - (U,W_1,W_2,X_k,S_k,Z_k) - (X_{\bar{k}}, S_{\bar{k}}, Z_{\bar{k}}, Y, V_{\bar{k}})$ for $k, \bar{k} \in \{1,2\}$ and $k\neq \bar{k}$.
\end{proof}

\begin{lemma}
    Let $(U,W_1,W_2,X_1,X_2,h) \in \calP_D^{\MAC}$ and $(R_1, R_2)\in\calR^{\MAC}(D)$. If $I(X_k; S_k, Y | X_{\bar{k}}) = 0$ for $k, \bar{k} \in \{1,2\}$ and $k\neq \bar{k}$, we have $R_k = 0$ and $R_{\bar{k}} < I(X_{\bar{k}}; Y)$.
\end{lemma}
\begin{proof}
    Without loss of generality, it is assumed $I(X_2; S_2, Y | X_1) = 0$. It shows that
    \begin{subequations}
    \begin{align}
        R_2 &< I(X_2, V_2; Y | U,W_1,W_2,X_1,V_1) + I(W_2; Z_1 | S_1,U,W_1,X_1) - I(V_2; S_2, Z_2 | U, W_1, W_2, X_2)\\
        &=I(X_2; Y | U,W_1,W_2,X_1,V_1) + I(V_2; Y | U,W_1,W_2,X_1,X_2,V_1) \notag\\
        &\hspace{3cm}+ I(W_2; Z_1 |S_1, U,W_1,X_1) - I(V_2; Z_2 | S_2, U, W_1, W_2, X_2)\\
        &= I(X_2; Y | U,W_1,W_2,X_1,V_1) + H(V_2 |  U,W_1,W_2,X_1,X_2,V_1) - H(V_2 |  U,W_1,W_2,X_1,X_2,V_1,Y) \notag\\
        &\hspace{2cm} + I(W_2; Z_1 |S_1, U,W_1,X_1) - H(V_2 | U, W_1, W_2, X_2) + H(V_2 | U, W_1, W_2, X_2,S_2,Z_2)\\
        &\le I(X_2; Y | U,W_1,W_2,X_1,V_1) + I(W_2; Z_1 |S_1, U,W_1,X_1)\label{appsubeq:mac-lemma2-1}\\
        &\le I(X_2; S_1, Y | U,W_1,W_2,X_1,V_1)+ I(W_2; S_1, Y | U,W_1,X_1)\label{appsubeq:mac-lemma2-2}\\
        &= H(S_1, Y | U,W_1,W_2,X_1,V_1) - H( S_1, Y | U,W_1,W_2,X_1,X_2,V_1) + H(S_1,Y | U,W_1,X_1) - H(S_1,Y | U,W_1,W_2,X_1)\\
        &= H(V_1,S_1, Y | U,W_1,W_2,X_1) - H(V_1|U,W_1,W_2,X_1) - H(V_1,S_1, Y | U,W_1,W_2,X_1,X_2) + H(V_1 | U,W_1,W_2,X_1,X_2)\notag\\
        &\hspace{2cm}+ H(S_1,Y | U,W_1,X_1) - H(S_1,Y | U,W_1,W_2,X_1)\\
        &\le H(V_1,S_1, Y | U,W_1,W_2,X_1) - H(V_1,S_1, Y | U,W_1,W_2,X_1,X_2) + H(S_1,Y | U,W_1,X_1) - H(S_1,Y | U,W_1,W_2,X_1)\\
        &= H(S_1,Y |U,W_1,W_2,X_1) + H(V_1 | U,W_1,W_2,X_1,S_1, Y) - H(S_1, Y | U,W_1,W_2,X_1,X_2)\notag\\
        &\hspace{2cm}- H(V_1 | U,W_1,W_2,X_1,X_2,S_1, Y)+ H(S_1,Y | U,W_1,X_1) - H(S_1,Y | U,W_1,W_2,X_1)\\
        &= - H(S_1, Y | U,W_1,W_2,X_1,X_2) + H(S_1,Y | U,W_1,X_1)\label{appsubeq:mac-lemma2-3}\\
        &\le  H(S_1,Y | X_1) - H(S_1, Y | X_1,X_2)\label{appsubeq:mac-lemma2-4}\\
        &= I(X_2; S_1,Y | X_1),
    \end{align}
    \end{subequations}
    where \eqref{appsubeq:mac-lemma2-1} holds because conditions do not increase entropy and $V_2 - (U,W_1,W_2,X_2,S_2,Z_2) - (X_1,V_1,Y)$ forms a Markov chain, \eqref{appsubeq:mac-lemma2-2} holds since $Z_1$ is a deterministic function of $Y$, \eqref{appsubeq:mac-lemma2-3} follows the Markov chains $V_1-(U,W_1,W_2,X_1,X_2,S_1,Z_1) - (X_2,Y)$, \eqref{appsubeq:mac-lemma2-4} results from the Markov chain $(S_1,Y) - (X_1,X_2) - (U,W_1,W_2)$. Thus, we have
    \begin{subequations}
    \begin{align}
        R_1 & = R_1 + R_2\\
        &<  I(X_1,X_2;Y)\\
        & =  I(X_1;Y) + I(X_2; Y |X_1)\\
        &\le I(X_1;Y) + I(X_2; S_1,Y |X_1)\\
        &=I(X_1;Y).
    \end{align}
    \end{subequations}
\end{proof}

Consequently, with the above lemmas, we can assume $I(X_1;Y|X_2) > \mu_1 >0$ and $I(X_2; S_1,Y | X_1) > \mu_2 >0$ in the following. Otherwise, the channel becomes equivalent to a point-to-point channel, and the proof of achievability is trivial. Fixing $\calP_D^{\MAC}$ that achieves $\calR^{\MAC}(D)$, we assume the transmission happens in $B+4$ blocks.

\subsection{Codebook Generation}

For each block $b$, split the private messages into to two parts $m_{k,b} = (m_{k,b}', m_{k,b}'')$ with $m_{k,b}'\in [2^{nR_k'}]$, $m_{k,b}''\in [2^{nR_k''}]$ and $R_k = R_k' + R_k''$ for Encoder $k\in \{1,2\}$. Let $l_{k,b}\in [2^{nR_{sk}}]$ and $m_{c,b} = (m'_{1,b-1}, m'_{2,b-1})$.

For block $b\in [B+1]$, Encoder $k$ generates 
\begin{itemize}
    \item $2^{n(R_1'+R_2')}$ sequences of $u^n(m_{c,b})$ \ac{iid} according to $P^n_U$;
    \item $2^{nR'_k}$ sequences of $w_k^n(m'_{k,b} | m_{c,b})$ \ac{iid} according to $P^n_{W_k|U}$ given each $u^n(m_{c,b})$;
    \item $2^{n(R_k'' + R_{sk})}$ sequences of $x_k^n(m_{k,b}'', l_{k,b-1} |m_{k,b}', m_{c,b})$ \ac{iid} according to $P^n_{X_k|W_kU}$ given each $w_k^n(m'_{k,b} | m_{c,b})$ and $u^n(m_{c,b})$;
    \item $2^{nR_{sk}}$ sequences of $v_k^n(l_{k,b} | m_{k,b}'', m_{1,b}', m_{2,b}', l_{k,b-1}, m_{c,b})$ \ac{iid} according to $P^n_{V_k|UW_1W_2X_k}$ given each $u^n(m_{c,b})$, $w_1^n(m'_{1,b} | m_{c,b})$, $w_2^n(m'_{2,b} | m_{c,b})$ and $x_k^n(m_{k,b}'', l_{k,b-1} |m_{k,b}', m_{c,b})$.
\end{itemize}
 % Let $m_{c,b} = (m'_{1,b-1}, m'_{2,b-1})$, for each decoder $k$, generate $2^{nR'_k}$ sequences of $w_k^n(m'_{k,b} | m_{c,b})$ according to $P_{W_k|U}$, $2^{n(R_k'' + R_{sk})}$ sequences of $x_k^n(m_{k,b}'', l_{k,b-1} |m_{k,b}', m_{c,b})$ according to $P_{X_k|W_kU}$, $2^{nR_{sk}}$ sequences of $v_k^n(l_{k,b} | m_{k,b}'', m_{1,b}', m_{2,b}', l_{k,b-1}, m_{c,b})$ according to $P_{V_k|UW_1W_2X_k}$.

For block $B+2$, generate one length-$n_1$ sequence $x_2^{n_1}$ \ac{iid} from $P^n_{X_2}$, and $2^{nR_{s1}}$ length-$n_1$ sequences $x_1^{n_1}(l_{1,B+1})$ \ac{iid} from $P^n_{X_1}$ with $n_1 = \frac{nR_{s1}}{\mu_1}$.

For block $B+3$, generate one length-$n_2$ sequence $x_1^{n_2}$ \ac{iid} from $P^n_{X_1}$, and $2^{nR_{s2}}$ length-$n_2$ sequences $x_2^{n_2}(l_{2,B+1})$ \ac{iid} from $P^n_{X_2}$ with $n_2 = \frac{nR_{s2}}{\mu_2}$.

For block $B+4$, generate one length-$n_3$ sequence $x_2^{n_3}$ \ac{iid} from $P^n_{X_2}$, and $2^{n_2(H(S_1) + \delta)}$ length-$n_3$ sequences $x_1^{n_3}(l_{1,B+3})$ \ac{iid} from $P^n_{X_1}$ with $l_{1,B+3}=[2^{n_2(H(S_1) + \delta)}]$ and $\delta>0$ with $n_3 = \frac{n_2(H(S_1) + \delta)}{\mu_1}$. 

Note that in the last three blocks, the sequences are generated independently from $P_{X_1}$ and $P_{X_2}$, respectively.

\subsection{Encoding}

Let $l_{k,0} = m_{k,0}'=1$ for both $k\in\{1,2\}$. In each block $b\in[B]$, Encoder $k$ transmits $x_k^n(m_{k,b}'', l_{k,b-1} |m_{k,b}', m_{c,b})$. At the end of each block, with the knowledge of $(s_k^n(b), z_k^n(b))$, it first decodes $\hat{m}'_{\bar{k},b}$ from the other Encoder $\bar{k}$ such that
\begin{equation}
    \br{w_{\bar{k}}^n(\hat{m}'_{\bar{k},b} | m_{c,b}), u^n(m_{c,b}), w_{k}^n(m'_{k,b} | m_{c,b}), x_k^n(m_{k,b}'', l_{k,b-1} |m_{k,b}', m_{c,b}),s_k^n(b), z_k^n(b)} \in \typset{}{P_{UW_1W_2X_kS_kZ_k}},
\end{equation}
which is shown to be correct if
\begin{equation}\label{appeq:mac-rq'}
    R_{\bar{k}}' < I(W_{\bar{k}}; W_k, X_k, S_k, Z_k | U) = I(W_{\bar{k}};  Z_k | S_k, U,W_k,X_k).
\end{equation}
The decoded $\hat{m}'_{\bar{k},b}$ is then used to construct $m_{c,b+1}$ for the next block. In addition, Encoder $k$ then looks for a unique state description $\hat{l}_{k,b}$ such that 
\begin{equation}
\begin{split}
    &\left (v_k^n(\hat{l}_{k,b} | m_{k,b}'', m_{1,b}', m_{2,b}', l_{k,b-1}, m_{c,b}), u^n(m_{c,b}), \right.\\
    &\hspace{2cm}\left. w_1^n(m'_{1,b} | m_{c,b}), w_2^n(m'_{2,b} | m_{c,b}), x_k^n(m_{k,b}'', l_{k,b-1} |m_{k,b}', m_{c,b}), s_k^n(b), z_k^n(b)\right) \in \typset{}{P_{V_kUW_1W_2X_kS_kZ_k}}.
\end{split}
\end{equation}
If there are multiple such indices, select one from them randomly. If there is no such index, set $\hat{l}_{k,b}=1$.
Similarly as before, the existence of such $\hat{l}_{k,b}$ is guaranteed if
\begin{equation}\label{appeq:mac-rsq}
    R_{sk} > I(V_k; S_k,Z_k | U, W_1, W_2, X_k).
\end{equation}

In block $B+1$, Encoder $k$ sends $x_k^n(m_{k,B+1}'', l_{k,B} |1, m_{c,B})$ without encoding the shared message for the other encoder. At the end of block $B+1$, it thus only needs to find the state description $l_{k,B+1}$ as before.

In block $B+2$, Encoder 1 only sends $l_{1,B+1}$ without any fresh message, i.e., $x_1^{n_1}(l_{1,B+1})$. Encoder 2 sends a deterministic sequence $x_2^{n_1}$.

In block $B+3$, Encoder 2 only sends $l_{2,B+1}$ without any fresh message, i.e., $x_2^{n_2}(l_{2,B+1})$. Encoder 1 sends a deterministic sequence $x_1^{n_2}$. At the end of this block, with the knowledge of $s_1^{n_2}(B+3)$, Encoder 1 compresses it losslessly into an index $l_{1,B+3}$. This raises no error due to the choice of $n_3$ and lossless source coding theorem\cite{el2011network}.

In block $B+4$, Encoder 1 only sends $l_{1,B+3}$ without any fresh message, i.e., $x_1^{n_3}(l_{1,B+3})$. Encoder 2 sends a deterministic sequence $x_2^{n_3}$.

\subsection{Decoding}

In block $B+4$, upon receiving $y^{n_3}(B+4)$, the decoder finds the unique index $\hat{l}_{1,B+3}$ such that
\begin{equation}
    \br{x_1^{n_3}(\hat{l}_{1,B+3}), x_2^{n_3},y^{n_3}(B+4)} \in \calT^{(n_3)}_{\epsilon}(P_{X_1X_2Y}),
\end{equation}
which is guaranteed to be correct due to
\begin{equation}
    \frac{n_2(H(S_1) + \delta)}{n_3} = \mu_1 < I(X_1;Y | X_2).
\end{equation}
With the correct decoding of $l_{1,B+3}$, the decoder can recover $s_1^{n_2}(B+3)$ losslessly, which is then used to decode $\hat{l}_{2,B+1}$ in block $B+3$ such that
\begin{equation}
    \br{x_2^{n_2}(\hat{l}_{2,B+1}), x_1^{n_2}, s_1^{n_2}(B+3),y^{n_2}(B+3)} \in \calT^{(n_2)}_{\epsilon}(P_{X_1X_2S_1Y}),
\end{equation}
guaranteed by
\begin{equation}
    \frac{nR_{s2}}{n_2} = \mu_2 < I(X_2; S_1,Y | X_1).
\end{equation}
In block $B+2$, the decoder finds a unique index $\hat{l}_{1,B+1}$ such that
\begin{equation}
    \br{x_1^{n_1}(\hat{l}_{1,B+1}), x_2^{n_1},y^{n_1}(B+2)} \in \calT^{(n_1)}_{\epsilon}(P_{X_1X_2Y})
\end{equation}
with 
\begin{equation}
    \frac{nR_{s1}}{n_1} = \mu_1 < I(X_1; Y |X_2).
\end{equation}
Having the correct indices $(l_{1,B+1}, l_{2,B+1})$, the decoder then looks for $(\hat{m}_{c,B+1}, \hat{m}_{1,B+1}'', \hat{m}_{2,B+1}'', \hat{l}_{1,B}, \hat{l}_{2,B})$ such that
\begin{equation}
\begin{split}
    &\left(u^n(\hat{m}_{c,B+1}), w_1^n(1 | \hat{m}_{c,B+1}), w_2^n(1|\hat{m}_{c,B+1}), x_1^n(\hat{m}_{1,B+1}'',\hat{l}_{1,B} | \hat{m}_{c,B+1},1), x_2^n(\hat{m}_{2,B+1}'',\hat{l}_{2,B} | \hat{m}_{c,B+1},1),\right. \\
    &\hspace{0.5cm} \left. v_1^n(\hat{l}_{1,B+1} | \hat{m}_{1,B+1}'',1,1,\hat{l}_{1,B}, \hat{m}_{c,B+1}), v_2^n(\hat{l}_{2,B+1} | \hat{m}_{2,B+1}'',1,1,\hat{l}_{2,B}, \hat{m}_{c,B+1}), y^n(B+1) \right ) \in \typset{}{P_{UW_1W_2X_1X_2V_1V_2Y}},
\end{split}
\end{equation}
and, if correct, decodes $(\hat{m}_{c,b}, \hat{m}_{1,b}'', \hat{m}_{2,b}'', \hat{l}_{1,b-1}, \hat{l}_{2,b-1})$ with
\begin{equation}
\begin{split}
    &\left(u^n(\hat{m}_{c,b}), w_1^n(m_{1,b}' | \hat{m}_{c,b}), w_2^n(m_{2,b}'|\hat{m}_{c,b}), x_1^n(\hat{m}_{1,b}'',\hat{l}_{1,b-1} | \hat{m}_{c,b},m_{1,b}'), x_2^n(\hat{m}_{2,b}'',\hat{l}_{2,b-1} | \hat{m}_{c,b},m_{2,b}'),\right. \\
    &\hspace{1cm} \left. v_1^n(l_{1,b} | \hat{m}_{1,b}'',m_{1,b}',m_{2,b}',\hat{l}_{1,b-1}, \hat{m}_{c,b}), v_2^n(l_{2,b} | \hat{m}_{2,b}'',m_{1,b}',m_{2,b}',\hat{l}_{2,b-1}, \hat{m}_{c,b}), y^n(b) \right ) \in \typset{}{P_{UW_1W_2X_1X_2V_1V_2Y}},
\end{split}
\end{equation}
in the backward direction from $b=B$ to $b=1$. The probabilities of error events 
\begin{equation}
\begin{split}
    &\brcur{ (\hat{m}_{c,b}, \hat{m}_{\bar{k},b}'', \hat{l}_{\bar{k},b-1}) =( m_{c,b}, m_{\bar{k},b}'', l_{\bar{k},b-1}),\  (\hat{m}_{k,b}'', \hat{l}_{k,b-1}) \neq (m_{k,b}'', l_{k,b-1})},\\
    &\brcur{ \hat{m}_{c,b} = m_{c,b},\  (\hat{m}_{1,b}'', \hat{m}_{2,b}'', \hat{l}_{1,b-1}, \hat{l}_{2,b-1}) \neq (m_{1,b}'', m_{2,b}'', l_{1,b-1}, l_{2,b-1})},\\
    &\brcur{(\hat{m}_{c,b}, \hat{m}_{1,b}'', \hat{m}_{2,b}'', \hat{l}_{1,b-1}, \hat{l}_{2,b-1}) \neq (m_{c,b}, m_{1,b}'', m_{2,b}'', l_{1,b-1}, l_{2,b-1})}
\end{split}
\end{equation}
are bounded to $0$ as $n\to \infty$ if
\begin{equation}\label{appeq:mac-rq}
\begin{split}
    R_k'' + R_{sk} &< I(X_k,V_k; Y | U,W_1,W_2,X_{\bar{k}},V_{\bar{k}})\\
    R_1'' + R_2'' + R_{s1} + R_{s2} & < I(X_1,X_2,V_1,V_2; Y | U,W_1,W_2)\\
    R_1+R_2+R_{s1} + R_{s2} &< I(U,W_1,W_2,X_1,X_2,V_1,V_2; Y),
\end{split}
\end{equation}
and all the other sub-events are shown to be redundant.
The decoder then estimates the channel state at each block by applying the estimator
\begin{equation}
    \hat{s}_i = h(u_i, w_{1,i}, w_{2,i}, x_{1,i}, x_{2,i}, v_{1,i}, v_{2,i}, y_i)
\end{equation}
in each block.

\subsection{Analysis}

Combining \eqref{appeq:mac-rq'}, \eqref{appeq:mac-rsq} and \eqref{appeq:mac-rq}, we obtain
\begin{equation}
\begin{split}
    R_1 &< I(X_1, V_1; Y | U,W_1,W_2,X_2,V_2) + I(W_1; Z_2 |S_2, U,W_2,X_2) - I(V_1; S_1, Z_1 | U, W_1, W_2, X_1)\\
    R_2 &< I(X_2, V_2; Y | U,W_1,W_2,X_1,V_1) + I(W_2; Z_1 |S_1, U,W_1,X_1) - I(V_2; S_2, Z_2 | U, W_1, W_2, X_2)\\
    R_1 + R_2 &< I(X_1, X_2, V_1, V_2 ; Y | U,W_1,W_2) +  I(W_1; Z_2 | S_2, U,W_2,X_2)+ I(W_2; Z_1 |S_1, U,W_1,X_1)\\
    &\hspace{5cm}-  I(V_1; S_1, Z_1 | U, W_1, W_2, X_1)- I(V_2; S_2, Z_2 | U, W_1, W_2, X_2)\\
    R_1 + R_2&<  I(U,W_1,W_2,X_1,X_2,V_1,V_2; Y) - I(V_1; S_1, Z_1 | U, W_1, W_2, X_1)- I(V_2; S_2, Z_2 | U, W_1, W_2, X_2).
\end{split}
\end{equation}
With all the codewords correctly decoded, the averaged estimation distortions are also under $D$ as $n \to \infty$. The data rates for the three messages converge to $(R_1, R_2)$ as $B\to \infty$. Furthermore, using the time-sharing strategy, it shows that for any two achievable rate pairs $(R_{11}, R_{21}), (R_{12}, R_{22}) \in \calR_i^{\MAC}(D)$, their convex combination $(R_1, R_2) = (\alpha R_{11} + \bar{\alpha}R_{12}, \alpha R_{21} + \bar{\alpha} R_{22})$ with $0\le \alpha \le 1$ and $\bar{\alpha}=1-\alpha$ is also achievable. At the same time, the attained distortion doesn't change. Hence, the resulting achievable region is the convex hull of $\calR_i^{\MAC}(D)$. This concludes the proof.
\fi

% \section*{Acknowledgment}


\ifCLASSOPTIONcaptionsoff
  \newpage
\fi



%\bibliographystyle{IEEEtran}
\bibliographystyle{IEEEtran}
\bibliography{IEEEabrv,mybib}


% \begin{IEEEbiographynophoto}{Xinyang Li}

% \begin{IEEEbiographynophoto}{Vlad C. Andrei}

% \begin{IEEEbiographynophoto}{Aladin Djuhera}

% \begin{IEEEbiographynophoto}{Ullrich J. M\"onich}

% \begin{IEEEbiographynophoto}{Holger Boche}

\end{document}